\documentclass[11pt]{article}
\usepackage[hmargin={1.0in, 1.0in}, vmargin={1.0in, 1.0in}]{geometry}
\usepackage{amsmath}
\usepackage{amssymb}
\usepackage{mathtools}
\usepackage{float}
\usepackage{natbib}
\usepackage{epsfig}
\usepackage{enumitem}
\usepackage{bbm, dsfont}
\usepackage{graphicx}
\usepackage{adjustbox}
\usepackage{color,latexsym,amsfonts}
\usepackage[lined,boxed,commentsnumbered, ruled]{algorithm2e} 
\usepackage{bm}
\usepackage{multirow}
\usepackage{titlesec}
\usepackage{subfig}

\usepackage{chngcntr}

\usepackage{setspace}
\newtheorem{prop}{Proposition}
\newtheorem{theorem}{\bf Theorem}
\newtheorem{lemma}{Lemma}
\newtheorem{remark}{Remark}
\newtheorem{corollary}{Corollary}

\newcommand{\be}{\begin{array}}
\newcommand{\ee}{\end{array}}

\newcommand{\ignore}[1]{}{}

\newcommand{\bbeta}{{\mbox{\boldmath $\bbeta$}}}

\def\bed{\begin{description}}\def\eed{\end{description}}
\def\ben{\begin{enumerate}}\def\een{\end{enumerate}}
\def\bea{\begin{eqnarray}}\def\eea{\end{eqnarray}}
\def\bean{\begin{eqnarray*}}\def\eean{\end{eqnarray*}}
\def\ba{\begin{array}}\def\ea{\end{array}}
\def\bt{\begin{theorem}  {\bf }}\def\et{\end{theorem}}
\def\bl{\begin{lemma}{\bf }}\def\el{\end{lemma}}
\def\br{\begin{remark}$\!\!\!\!\!$ {\bf .}}\def\er{\end{remark}}
\def\bc{\begin{corollary}$\!\!\!\!\!$ {\bf }}\def\ec{\end{corollary}}

\newenvironment{proof}{\paragraph{Proof:}}{\hfill$\square$}

\newcommand{\beqn}{\begin{eqnarray}}             
\newcommand{\eeqn}{\end{eqnarray}}               
\newcommand{\beq}{\begin{eqnarray*}}             
\newcommand{\eeq}{\end{eqnarray*}}
\DeclareMathOperator*{\argmax}{arg\,max}
\DeclareMathOperator*{\argmin}{arg\,min}

\begin{document}
	
\setlength{\abovedisplayskip}{4pt}
\setlength{\belowdisplayskip}{4pt}
\setlength{\abovedisplayshortskip}{2pt}
\setlength{\belowdisplayshortskip}{2pt}

\title{Predicting Future Change-points in Time Series}

\author{Chak Fung Choi$^1$, Chunxue Li$^1$, Chun Yip Yau\thanks{Corresponding author: cy.yau@cuhk.edu.hk}$^{*1}$, Zifeng Zhao$^2$}
\date{$^1$ Department of Statistics, Chinese University of Hong Kong \\
			$^2$ Department of Business Analytics, University of Notre Dame}

\maketitle

\begin{abstract}
	Change-point detection and estimation procedures have been widely developed in the literature. However, commonly used approaches in change-point analysis have mainly been focusing on detecting change-points within an entire time series (off-line methods), or quickest detection of change-points in sequentially observed data (on-line methods). Both classes of methods are concerned with change-points that have already occurred. The arguably more important question of when future change-points may occur, remains largely unexplored. In this paper, we develop a novel statistical model that describes the mechanism of change-point occurrence. Specifically, the model assumes a latent process in the form of a random walk driven by non-negative innovations, and an observed process which behaves differently when the latent process belongs to different regimes. By construction, an occurrence of a change-point is equivalent to hitting a regime threshold by the latent process. Therefore, by predicting when the latent process will hit the next regime threshold, future change-points can be forecasted. The probabilistic properties of the model such as stationarity and ergodicity are established. A composite likelihood-based approach is developed for parameter estimation and model selection. Moreover, we construct the predictor and prediction interval for future change points based on the estimated model. 
\end{abstract}

keywords: Composite likelihood, Forecast, Maximum a posteriori estimation, State space model, Structural break, Threshold model.

\section{Introduction}

Change-point analysis plays an important role in many scientific disciplines including climate science \citep{Climate2012, Environment2016}, engineering \citep{Leo1950}, economics, genetics \citep{Muller2000}, neuroimaging \citep{Kaplan2000}, and many other fields \citep{Matteson2014, Cho2015, Enikeeva2019}. 

The change-point problem can be classified into two major categories: on-line problems and off-line problems. The goal of the on-line change-point problem is to monitor sequentially observed data in order to quickly detect a change-point upon its occurrence, which is important in many applications such as quality control \citep{Mei2006}. Both Bayesian and non-Bayesian approach have been developed for on-line change-point problems, see \citet{Ryan2007}, \citet{Choi2012} and \citet{Fearnhead2007}. On the other hand, the off-line problem, also called retrospective change-point analysis, aims to detect abrupt changes for an entirely observed time series. Testing the existence of change-points, and locating the change-points accurately and computational efficiently are of main concerns. For some comprehensive overviews, see \citet{Csorgo-Horvath97}, 
\citet{Aue-Horvath13} and \citet{Jandhyala-et-al13}. 

There is also literature related to prediction in change-point analysis. For example, \citet{Pesaran2004} proposed a method to forecast future observations under structural breaks, where given the estimated change-points, future values of the time series are predicted. Based on \citet{Pesaran2004}, \citet{Tian2014} considered a new weighting scheme for forecasting. 
However, these works focus on predicting future values of the time series given existing structural changes. To the best of our knowledge, no method has been proposed to predict the locations of future change-points. This motivates us to develop a model to describe the mechanism of change-point occurrence, which allows one to predict future change-point locations.

Specifically, we develop a multiple-regime threshold autoregressive state space (TASS) model for predicting the locations of future change-points. The TASS model connects the observed time series to an unobserved latent process. In particular, the observed time series follows a multiple-regime threshold model, and its regime classification is based on the latent process and a set of thresholds. The latent process is a random walk cycling across the unit interval. We show that the TASS model possesses desirable probabilistic properties such as stationarity and ergodicity. By construction, the occurrence of a change-point is equivalent to hitting a regime threshold by the latent process. Therefore, by studying the hitting time of the latent process to the next regime threshold, a future change-point can be predicted and inferred.

Due to the presence of the latent process, the proposed TASS model requires new estimation, inference and model selection procedures. In particular, the full likelihood involves large number of integrals and is computationally infeasible. To tackle this problem, we develop composite likelihood-based procedures for accurate and computationally efficient model estimation and selection. Moreover, a maximum a posteriori sequence estimation algorithm is further proposed to conduct inference on the latent process and thus enables the prediction of future change-points. 

The rest of the paper is organized as follows. Section 2 specifies the TASS model and studies its probabilistic properties. Section 3 proposes the composite likelihood based model estimation, selection procedure and model diagnostic checks. Section 4 discusses the prediction of future change-points. Sections 5 and 6 present simulation studies and real data analysis, respectively. Technical proofs of all propositions, and theorems are gathered in the Appendix.

\section{Multiple-regime Threshold Autoregressive State Space Model}

\subsection{Model specification} \label{sec:model.def}
Denote the observed time series by $\{X_1,\ldots,X_{n}\}$. We say that $\{X_t\}_{t=1}^{n}$ follows an $m$-regime threshold autoregressive state space (hereafter, TASS($m$)) model if
\begin{eqnarray}
	X_t &=&
	\begin{cases}
		a_1+\phi_1(X_{t-1}-a_1)+\sigma_1 e_t\,, &Y_t \in [0,r_1)\,, \\
		a_2+\phi_2(X_{t-1}-a_2)+\sigma_2 e_t\,, &Y_t \in [r_1,r_2)\,, \\
		... \\
		a_m+\phi_m(X_{t-1}-a_m)+\sigma_m e_t\,, &Y_t \in [r_{m-1},1)\,,
	\end{cases}\label{model} \\
	Y_t &=&
	\begin{cases}
		Y_{t-1}+ \epsilon_t\,, &\mbox{if } Y_{t-1}+\epsilon_t<1\,,  \\
		Y_{t-1}+ \epsilon_t - \lfloor Y_{t-1}+ \epsilon_t \rfloor \,,  &\mbox{otherwise} \,,
	\end{cases}  \label{Y.RW} 
\end{eqnarray}

where $\{Y_{t}\}_{t=1}^{n}$ is the latent process, $\epsilon_t \stackrel{iid}{\sim}  \mbox{Gamma}(\alpha, \beta)$, $e_t \stackrel{iid}{\sim} \mbox{N}(0,1)$, the innovations $\{\epsilon_t\}$ and $\{e_t\}$ are independent, and $0=r_0< r_1 < \cdots <r_{m-1}<r_m= 1$ are thresholds that partition $\{Y_{t}\}$ into $m$ regimes. Here, $\lfloor Y_{t-1}+ \epsilon_t \rfloor$ denotes the integer part of $Y_{t-1}+ \epsilon_t$. For $Y_t \in [r_{j-1}, r_j)$, i.e., $Y_t$ is in the $j$th regime, the observation $X_t$ follows an AR(1) model with mean $a_j$ and autoregressive coefficient $\phi_j$. Extension to higher order auto regression models is possible but we focus on AR(1) to facilitate presentation. The parameters of the TASS model include the number of regimes $m$, threshold values $r_j$, $j=1,\ldots,m-1$, parameters of the AR models $(\phi_i, a_i, \sigma_i)$, $i=1,\ldots,m$, and the parameters $\alpha$ and $\beta$ of the latent process. For a fixed $m$, denote the parameter vector and its parameter space as $\theta_m=( \phi_1, a_1, \sigma_1,\ldots,\phi_m, a_m, \sigma_m, r_1,\ldots,r_{m-1}, \alpha, \beta )$ and $\Theta_m$, respectively. We assume that the parameter space $\Theta_m$ is compact.

The latent process $\{Y_t\}$ in \eqref{Y.RW} is a partial sum process that accumulates an independent and identically distributed gamma random variable $\epsilon_t$ at each time point $t$. Since a gamma random variable is always positive, $Y_t$ increases and travels along successive regimes $[r_j,r_{j+1})$ until $Y_{\tau-1}+\epsilon_{\tau}$ hits 1 at some time point $\tau$, and in this case $Y_{\tau}$ will be restarted by subtracting the integer part of the partial sum. Thereafter, $\{Y_t\}$ will repeat this process within $[0,1]$. We are primarily interested in the case $E(\epsilon_{t})=\frac{\alpha}{\beta} \ll 1$ such that the latent process $\{Y_{t}\}$ does not experience regime-switch frequently. By construction, $\{Y_t\}$ is a Markov chain. 

The key idea of modeling and predicting change-points using the TASS model (i.e. model \eqref{model}) is that the observed time series $X_{t}$ experiences a change-point when the latent process $Y_{t}$ hits a regime threshold $r_j$, $j=1,\ldots,m$. Therefore, by predicting when $Y_t$ hits the next regime threshold, the location of a future change-point can be predicted. As is shown later, the random walk nature of $Y_{t}$ in \eqref{Y.RW} facilitates the computation of the regime threshold hitting time. 

\subsection{Comparisons with existing time series models} \label{sec:comparison}
The TASS model may appear similar to the well known self-excited multiple-regime threshold autoregressive (TAR) model and the hidden Markov model (HMM).

The TAR model, proposed by \citet{Tong1978}, is defined as
\begin{equation*}
	X_t=\begin{cases}
		a_1+\phi_1(X_{t-1}-a_1)+\sigma_1e_t\,, &X_{t-d} \in [r_0,r_1)\,, \\
		a_2+\phi_2(X_{t-1}-a_2)+\sigma_2e_t\,, &X_{t-d} \in [r_1,r_2)\,, \\
		... \\
		a_m+\phi_m(X_{t-1}-a_m)+\sigma_me_t\,, &X_{t-d} \in [r_{m-1},r_{m})\,.
	\end{cases}
\end{equation*} 
The HMM model assumes that the distribution of the observation $X_t$ depends on the current state of a hidden (latent) Markov process $C_t$; see for example, \citet{ZucchiniMacDonald2009}. The proposed TASS model is structurally different from the TAR model and the HMM model.

First, in the TAR model, the autoregressive model of the observation $X_t$ is governed by the values of the past observation $X_{t-d}$, for some $d>0.$ As $X_t$ may take a wide range of values in the real line, regime switching tends to occur frequently. In contrast, the latent process $\{Y_{t}\}$ in the TASS model is allowed to increase slowly and can stay in a regime for a fairly long period, thus is more suitable to model change-point data with piecewise stationary structures. Moreover, the random walk nature of the latent process $\{Y_{t}\}$ allows a simpler prediction theory than predicting future regime switches in the TAR model, as the stationary distribution of $X_{t-d}$ is difficult to be characterized in the TAR model.  

Second, in many applications of HMM, the transition probability of the latent Markov process is time homogeneous, i.e., the probability of the latent process changing from one state to another is constant across time. Moreover, the observation $X_t$ in HMM are conditionally independent given the hidden state. In contrast, in the TASS model, due to the positive increment design of the latent process $Y_t$, the transition probability of $Y_t$ across regimes is not time homogeneous. In particular, the longer $Y_t$ stays in the same regime, the more likely it will breach the regime boundary, causing both the latent process $Y_t$ and the observation $X_t$ to switch to the next regime. Also, the observation $X_t$ follows different autoregressive models in different regimes, which is neither conditionally independent nor depending only on the latent process $Y_t$. This property demonstrates the flexibility of the TASS model and distinguishes it from HMM. Furthermore, in the TASS model, the location of the states $[r_j,r_{j+1})$s are explicitly modeled, in contrast to HMM where the states are only assumed hidden. 

\subsection{Probabilistic properties} \label{properties}

In this subsection, we study some probabilistic properties of the TASS model, which are useful for parameter estimation and future change-point prediction.

For the TASS model, the conditional distribution of $X_t$ given $X_{t-1}$ ,$\cdots, X_1$ and $Y_t,\cdots,Y_1$ is equivalent to the conditional distribution of $X_t$ given $X_{t-1}$ and $Y_t$. In addition, the conditional distribution of $Y_t$ given $Y_{t-1},...,Y_1$ is equal to the distribution of $Y_t$ given $Y_{t-1}$. To derive $p(y_t \vert y_{t-1})$, note from \eqref{Y.RW} that given $Y_{t-1}=y_{t-1}$, $Y_{t}$ takes value $y_{t}\in(y_{t-1},1)$ if $\epsilon_t=y_{t}-y_{t-1}+j$ for $j=0,1,\ldots$, or $y_{t}\in[0,y_{t-1}]$ if $\epsilon_t=y_{t}-y_{t-1}+j$ for $j=1,2\ldots$ Therefore, 

\begin{eqnarray}
	p(y_t \vert y_{t-1}) = g_{\alpha,\beta}(y_t-y_{t-1})\mathbbm 1_{\{y_{t}>y_{t-1}\}} + \sum_{j=1}^{\infty}g_{\alpha,\beta}(y_t-y_{t-1}+j)\,, \ \ \ \ \mbox{for } 0 \leq y_t <1\, ,\label{conditionpr2}
\end{eqnarray} 

where $g_{\alpha,\beta}(x)={\beta^\alpha}/{\Gamma(\alpha)}x^{\alpha-1}e^{-\beta x}$ is the density of $\mbox{Gamma}(\alpha,\beta)$ distribution and $\mathbbm{1}_{\{\cdot\}}$ is the indicator function.

Theorem \ref{stationary} gives the stationary distribution of $Y_t$. The result goes beyond the TASS model as it does not require $\epsilon_t$ to be Gamma distributed. 

\begin{theorem} \label{stationary}
	
	If the latent process $\{Y_t\}$ follows \eqref{Y.RW}, where $\{\epsilon_t\}$ are i.i.d. random variables with a continuous probability density support on $[0,\infty)$, then the stationary distribution of $Y_t$ is Uniform $(0,1)$. In other words, the invariant measure of the Markov chain is $\pi(y_t)=1$ with $y_t\in (0,1)$. 
	
\end{theorem}

Theorem \ref{ergodic} establishes the ergodicity and the $\beta$-mixing property for the TASS model via the theory of Markov chain as $\{(X_t,Y_t)\}$ forms a first-order Markov chain.

\begin{theorem} \label{ergodic}
	
	If the times series $\{X_t\}$ follows the TASS model with $|\phi_j| < 1, j =1,\ldots,m$, then $\{(X_t, Y_t)\}$ is a stationary and geometrically ergodic Markov chain, and is $\beta$-mixing with coefficient $\beta_t = O(\bar\gamma^t)$ for some 
	$\bar{\gamma} \in (0,1)$. 
	
\end{theorem}

The following lemma provides several conditional probabilities that will be useful for deriving the composite likelihood function in Section \ref{sec.est}. 

\begin{lemma} \label{prob.property}
	
	Suppose that $Y_t$ is in the $i$th regime and $Y_{t+1}$ is in the $j$th regime. Then,  
	the conditional probability density functions of $X_t$ given $Y_t$, and $X_{t+1}$ given 
	$X_t, Y_{t+1}$, are given respectively by 
	
	\begin{eqnarray} \label{x1}
		p(x_t \vert y_t)&=&\sqrt{\frac{1-\phi_i^2}{2\pi\sigma_i^2}}\exp\left(-\frac{(1-\phi_i^2)(x_t-a_i)^2}{2\sigma_i^2}\right)\,, \\
		p(x_{t+1} \vert x_t, y_{t+1})&=&\frac{1}{\sqrt{2\pi\sigma_j^2}}\exp\left(- \frac{(x_{t+1}-a_j-\phi_j(x_t-a_j))^2}{2\sigma_j^2} \right)\,.\label{x2}
	\end{eqnarray}
	
\end{lemma}

\section{Model Estimation and Inference}\label{sec.est}

\subsection{Consecutive tuple log-likelihood}\label{sec.tuple.lik}

Given that $x_{1:n}=(x_{1},x_2,\ldots,x_n)$ follows the TASS($m$) model, the full likelihood function can be expressed as 

\begin{eqnarray}
	p(x_{1:n};\theta_m) =\int\cdots\int p(x_1 \vert y_1)\left[\prod_{i=2}^{n}p(x_i \vert x_{i-1},y_i)\right] p(y_1) \left[ \prod_{i=2}^{n}p(y_i \vert y_{i-1})\right] dy_1\ldots dy_n \,. \label{full.lik}
\end{eqnarray}

The latent process $\{Y_t\}$ introduces multiple integrals which are computationally infeasible. 
To tackle this problem, we propose the consecutive tuple log-likelihood (CTL) estimator, a special case of composite likelihood which is particularly useful for time series; see \citet{Davis-Yau11}. 
Specifically, the Consecutive $k$-tuple Log-likelihood (CTL$_k$) is defined as

\begin{eqnarray*}
	\mbox{CTL}_k(\theta_m)=\sum_{t=1}^{n-k} \log p(x_t,x_{t+1},\ldots, x_{t+k};\theta_m)\,.
\end{eqnarray*}   

While computing the joint probability density function of $x_{1:n}$ is infeasible, the CTL$_k$ uses the sum of the computationally simpler joint probability density functions for $k$-tuples of observations.

Following the suggestion of \citet{Davis-Yau11}, we can gain the most computation efficiency by choosing $k$ as the smallest integer such that the model is identifiable with $p(x_t,x_{t+1},\ldots, x_{t+k};\theta_m)$. Some straightforward but tedious algebra shows that the TASS model is not identifiable with CTL$_1$. Therefore, for the estimation of the TASS model, we propose to use the CTL$_2$ function 

\begin{eqnarray} \label{CTL2}
	\mbox{CTL}_2(\theta_m)=\sum_{t=1}^{n-2} \log p(x_t,x_{t+1}, x_{t+2};\theta_m)\,.
\end{eqnarray}

The CTL$_2$ estimator can be found by maximizing \eqref{CTL2}. 
The joint densities in \eqref{CTL2} involves a three-dimensional integral, 
a computationally feasible formula of which is obtained in the next subsection based on the specific structure of the TASS model. 

\subsection{Consecutive tuple log-likelihood estimator}\label{sec:cpl}

To highlight the main idea and for the ease of presentation, we derive the CTL$_2$ function for the following two-regime TASS(2) model:

\begin{align*}
	X_t & =
	\begin{cases}
		a_1+\phi_1(X_{t-1}-a_1)+\sigma_1 e_t\,, &Y_t \in [0,r_1)\,, \\
		a_2+\phi_2(X_{t-1}-a_2)+\sigma_2 e_t\,, &Y_t \in [r_1,1) \,,
	\end{cases} \\ 
	Y_t & = 
	\begin{cases}
		Y_{t-1}+ \epsilon_t\,, &\mbox{if } Y_{t-1} + \varepsilon_t<1,  \\
		Y_{t-1}+ \epsilon_t-\lfloor Y_{t-1}+ \epsilon_t \rfloor \,,  &\mbox{otherwise},
	\end{cases}
\end{align*}

where  $\epsilon_t \overset{i.i.d.}{\sim} \mbox{Gamma}(\alpha, \beta)$, $e_t \overset{i.i.d.}{\sim} \mbox{N}(0,1)$. To derive the CTL$_2$ function in \eqref{CTL2}, note that 

\begin{eqnarray} \label{joint}
	p(x_t, x_{t+1}, x_{t+2}; \theta_2) =\int p(x_t, y_t) \left[ \prod_{i=1}^{2}p(x_{t+i} \vert x_{t+i-1}, y_{t+i})p(y_{t+i} \vert y_{t+i-1}) \right]
	dy_{t+2} dy_{t+1} dy_{t} \,.
\end{eqnarray}

Combining \eqref{joint} with the conditional densities 
\eqref{conditionpr2}, \eqref{x1} and \eqref{x2} established in Section \ref{properties}, we have the following 
formula for computing the joint density $p(x_t, x_{t+1}, x_{t+2};\theta_2)$. 

\begin{prop}\label{prop:triple.p} 
	The joint density of $p(x_t, x_{t+1}, x_{t+2};\theta_2)$ can be expressed as
	
	\begin{eqnarray}\label{joint.closed}
		p(x_t, x_{t+1},x_{t+2};\theta_2)=\sum_{i,j,k=1,2} g_t(i,j,k)w_t(i,j,k)\,,
	\end{eqnarray}
	
	where with $\mu_{1t}=\phi_j(x_{t}-a_j)+a_j$, $\mu_{2t}=\phi_k\phi_j(x_t-a_j)+\phi_k(a_j-a_k)+a_k$, $r_0=0$ and $r_2=1$,  
	
	\begin{eqnarray*}
		g_t(i,j,k)&=&\Phi_{x_t}\left(a_i,\frac{\sigma_i^2}{1-\phi_i^2} \right)
		\Phi_{x_{t+1}}\left( \mu_{1t},\sigma_j^2\right) 
		\Phi_{x_{t+2}}\left( \mu_{2t},\phi_k^2 \sigma_j^2+\sigma_k^2 \right)\,,\\
		w_t(i,j,k)&=&\int_{r_{i-1}}^{r_i} \int_{r_{j-1}}^{r_j} \int_{r_{k-1}}^{r_k} p(y_{t+2} \vert y_{t+1})p(y_{t+1} \vert y_t)p(y_t)dy_{t+2} dy_{t+1} dy_{t}\,,
	\end{eqnarray*}  
	
	where $\Phi_{x}(\mu, \sigma^2)$ 
	is the density function of the normal distribution with mean $\mu$ and variance $\sigma^2$. The detailed formulas of $w_t(i,j,k)$ are provided in equations \eqref{w.1} to \eqref{w.8} in the Appendix. 
\end{prop}

\begin{remark}\label{rm.CTL.m}
	When the number of regimes $m$ is larger than 2, the CTL$_2$ function can still be used for estimation. 
	In particular, with the same factorization method, CTL$_2$ can be expressed as
	
	\begin{align*}
		\mbox{CTL}_2(\theta_m) &= \sum_{t=1}^{n-2} \log p(x_t, x_{t+1}, x_{t+2};\theta_m) \\
		&=\sum_{t=1}^{n-2} \log\left( \sum_{i,j,k=1,\ldots, m} g_t(i,j,k) w_t(i,j,k)\right)\,,
	\end{align*}

where $g_t(i,j,k)$ is defined in \eqref{joint.closed} and $w_t(i,j,k)$ is defined analogous to \eqref{w.1}-\eqref{w.8} in the Appendix.
\end{remark} 

\begin{remark}
To avoid the label switching problem, we introduce the identification condition $a_1= \min_{i=1,\ldots,m} {a_i}$. Also, by inspecting equations \eqref{w.1} to \eqref{w.8} in the Appendix, it can be seen that 
$p(x_t, x_{t+1}, x_{t+2};\theta_m)$ is differentiable w.r.t.\ the threshold parameter $r_1$. Thus, unlike self-excited threshold models which encounter computational difficulties (see, e.g., \cite{Yau-et-al14}), the TASS model does not suffer from computational issues for optimizing the likelihood function $CTL_k$. 
\end{remark}

\subsection{Model selection}\label{sec:model.select}
In general, the estimation of the number of regimes $m$ can be regarded as a model selection problem, see \citet{Yau-et-al14} and \citet{ChanYauZhang2015} in the context of TAR models. For simplicity, we choose the Bayesian Information Criteria (BIC), which is widely used in the literature, to select the number of regimes. Additionally, we assume there exists an arbitrarily small constant $\epsilon>0$ such that $r_i-r_{i-1}>\epsilon$ for $i=1,\cdots,m$, which then implies an upper bound $M$ on the number of possible states of the TASS model. For $m=1,2,\cdots,M$, the BIC is defined as 

\begin{equation*}
	\mbox{BIC}(m) = (4m+2)\log n-\frac{2}{C_{2}}\mbox{CTL}_{2}(\hat{\theta}_m)\,,
\end{equation*}
where $\hat{\theta}_m=\argmax_{\theta_m}\mbox{CTL}_{2}(\theta_m)$, $4m+2$ is the total number of parameters, and 
$C_{2}=3(n-2)/n$ is the average number that an observation is used in the expression of $\mbox{CTL}_2$. The constant $C_2$ is used to make adjustment such that the order of $\mbox{CTL}_2$ is comparable to that of the full likelihood, see \citet{MaYau2016} for details. The estimated number of regimes $\hat{m}$ is then defined as the minimizer of $\mbox{BIC}(m)$ such that $\hat{m}=\argmin_{m} \mbox{BIC}(m).$ We have the following results on the consistency of parameter estimation and model selection. 

\begin{theorem} \label{thm.consistency}
	Suppose that the time series $\{X_t\}$ follows the TASS$(m_o)$ model with the true model parameters $\theta^{o}= ( \phi_1, a_1, \sigma_1, \ldots,\phi_{m_o}, a_{m_o}, \sigma_{m_o}, r_1,\ldots,r_{{m_o}-1},$ 
	$ \alpha_o, \beta_o )$. 
	Then, we have 
	\begin{align*}
		\hat{m} \overset{p}{\to} m_{o}\,, \ 
		\left|\hat{\theta}_{\hat{m}}-{\theta}^{o} \right| \stackrel{p}{\rightarrow} 0 \,,
		\text{ and } \ \sqrt{n}(\hat{\theta}_{\hat{m}}-{\theta}^{o}) \stackrel{d}{\rightarrow} N(0, \Sigma_2\Sigma_1^{-1}\Sigma_2)\,,
	\end{align*}
	where $\Sigma_1=E\left(\frac{\partial^2 \log p\left(x_1,x_2,x_3;\theta^{o}\right)}{\partial \theta_{m_o}\partial \theta_{m_o}^\top}\right)$ and $\Sigma_2=\mbox{Var}\left(\frac{\partial \log p\left(x_1,x_2,x_3;\theta^{o}\right)}{\partial \theta_{m_o}}\right)$.
\end{theorem}

\subsection{MAP sequence estimation and change-point detection}\label{sec:MAP}

The composite likelihood based procedure only provides a point estimate for the model parameter $\theta_{m}$, including the thresholds $r_1,\ldots,r_{m-1}$, but not for the location of the change-points, i.e., the time points when the latent process hits the thresholds. Nevertheless, if the latent process $\{Y_t\}$ can be estimated, then the locations of change-points can be readily obtained by comparing the estimated latent process and the estimated thresholds. In this section, we investigate the estimation of $\{Y_t\}$ by the maximum a posteriori (MAP) sequence estimation; see, for example, \citet{bootstrapfilter2013}. 

In MAP sequence estimation, we find the sequence of state values $Y_t=y_t, t=1,\ldots,n$, that corresponds to the point of highest probability density conditional on the observed data $X_t=x_t, t=1,\ldots,n$. In other words,
denote $\Omega=[0,1)$ as the sample space of a single $Y_t , t=1, \ldots, n$, and $\Omega^n$ as the sample space of $Y_{1:n}$, we estimate $\hat{y}_{1:n}$ from
\begin{eqnarray} \label{maximization}
	\hat{y}_{1:n} =\argmax_{y_{1:n} \in \Omega^n} \left\{p(y_{1:n} \vert x_{1:n})\right\}\,.         
\end{eqnarray}

The above optimization problem aims at searching the maximizer of $p(y_{1:n}\vert $ $x_{1:n})$ over a high dimensional continuous state space, which is extremely difficult for a non-Gaussian state space model (\citet{MAP2001}, \citet{bootstrapfilter2013}). On the other hand, if this high dimensional maximization problem is imposed over a discrete state space, it can be solved by some existing computation algorithms such as the Viterbi algorithm \citep{Viterbi1967,ZucchiniMacDonald2009}.

Therefore, we can obtain an approximate solution of \eqref{maximization} by finding a suitable discretization of $\Omega^n$ and then approximate \eqref{maximization} by a high dimensional maximization problem over the discrete state space. To find a suitable discretization of $\Omega^n$ for maximizing $p(y_{1:n} \vert x_{1:n})$, a natural way is to generate realizations of $y_{1:n}$ from the conditional density $p(y_{1:n} \vert x_{1:n})$. When these realizations, denoted as $y_{1:n}^{(i)}, i=1,\ldots,N$, are generated for a large enough number $N$, the set $\{y_{1:n}^{(i)}\in \Omega^{n}: i=1,\ldots,N\}$ can serve as a suitable discretization of $\Omega^n$.

Motivated by the problem of solving high dimensional continuous maximization problem through discretization, \citet{bootstrapfilter2013} proposed a numerical approach to generate discretization of $\Omega^n$ from $p(y_{1:n} \vert x_{1:n})$ by the particle filter algorithm. In particular, \citet{bootstrapfilter2013} applied one of the most well-known particle filter algorithms, the bootstrap filter \citep{Gordon1993}, to conduct the discretization of $\Omega^n$. 
For completeness, we summarize the details in Algorithm \ref{bootstrap.filter.algorithm}. 

\begin{algorithm}
	\caption{Bootstrap Filtering Algorithm for the discretization of $\Omega^n$ \citep{bootstrapfilter2013} }
	\label{bootstrap.filter.algorithm}
	\begin{tabbing}
		\enspace Initialization: draw $N$ i.i.d. samples $y_1^{(i)}, i=1,\ldots,N$ from initial distribution $p(y_1)$\\
		\enspace For $t=2$ to $t=n$ \\
		\qquad For $i=1$ to $i=N$ \\
		\qquad \qquad Draw $\bar{y}_t^{(i)}$ in $\Omega$ independently from $p(y_t \vert y_{1:t-1}^{(i)})=p(y_t\vert y_{t-1}^{(i)}) $ \\
		\qquad \qquad Set $\bar{y}_{1:t}^{(i)}=\{y_{1:t-1}^{(i)},\bar{y}_t^{(i)}\}$ \\
		\qquad \qquad Evaluate the importance weights $\tilde{w}_t^{(i)}=p(x_t \vert \bar{y}_{1:t}^{(i)},x_{1:t-1})=p(x_t \vert x_{t-1},\bar{y}_t^{(i)})$ \\
		\qquad For $i = 1$ to $i = N$ \\
		\qquad \qquad Normalize the importance weights to obtain $w_t^{(i)}={\tilde{w}_t^{(i)}}\big/{\sum_{k=1}^N\tilde{w}_t^{(k)}}$\\
		\qquad For $i=1$ to $i=N$ \\
		\qquad \qquad resampling: set $y_{1:t}^{(i)}=\bar{y}_{1:t}^{(k)}$ with probability $w_t^{(k)}$, $k \in \{1,\ldots,N\}$
	\end{tabbing}
	\vspace*{-10pt}    
\end{algorithm}

In Algorithm \ref{bootstrap.filter.algorithm}, we set the stationary distribution of $\{Y_t\}$ as the initial distribution $p(y_1)$. The transition probability  $p(y_t | y_{1:t-1}^{(i)})=p(y_t|y_{t-1}^{(i)})$  can be calculated based on \eqref{conditionpr2}.
The resulting random samples $\{y_{1:n}^{(i)}\}_{i=1,\ldots,N }$ are called particles. 

Next, with the discretized state space $\Omega^n_N:=\{y_{1:n}^{(i)},i=1,\ldots,N\}$, 
the MAP sequence estimation approximates the maximization in \eqref{maximization} by 
\begin{eqnarray} \label{max.appro}
	\hat{y}_{1:n}^N = \argmax_{y_{1:n} \in \Omega^n_N} p(y_{1:n} \vert x_{1:n})\,.
\end{eqnarray}
With finitely many elements in $\Omega^n_N$, the maximization problem in \eqref{max.appro} can be analytically solved by searching for the maximum value of $p(y_{1:n}^{(i)} \vert x_{1:n})$ over $i=1,\ldots,N$. To compute $p(y_{1:n}^{(i)} \vert x_{1:n})$, observe that
\begin{eqnarray} \label{prportion}
	p(y_{1:n}^{(i)} \vert x_{1:n})&=& \frac{p(y_{n}^{(i)} \vert y_{1:n-1}^{(i)})p(x_{n} \vert x_{1:n-1},y_{n}^{(i)})}{p(x_{n} \vert x_{1:n-1})}p(y_{1:n-1}^{(i)} \vert x_{1:n-1}) \nonumber \\
	& \propto & p(y_{n}^{(i)} \vert y_{n-1}^{(i)})p(x_{n} \vert x_{n-1},y_{n}^{(i)})p(y_{1:n-1}^{(i)} \vert x_{1:n-1})\,.
\end{eqnarray}

Hence, by defining $a_1^{(i)}=p(y_1^{(i)} \vert x_1)$ and $a_t^{(i)}=p(y_t^{(i)} \vert y_{t-1}^{(i)})p(x_t \vert x_{t-1},y_t^{(i)})a_{t-1}^{(i)}$ for $t=2,\ldots,n$, $p(y_{1:n}^{(i)} \vert x_{1:n})=Ca_n^{(i)}$ can be calculated recursively from $t=1,\ldots,n$, where $C$ is a constant. The knowledge of $C$ is not needed as $\hat{y}_{1:n}^N$ in \eqref{max.appro} equals to $y_{1:n}^{(N_0)}$, where $N_0= \argmax_{i \in \{1,\ldots,N\}} a_n^{(i)}$. The MAP sequence estimation algorithm is summarized in Algorithm \ref{MAP.algorithm}. Algorithm \ref{MAP.algorithm} incorporates Algorithm \ref{bootstrap.filter.algorithm}. Thus, in practice it suffices to conduct Algorithm \ref{MAP.algorithm} once.  

Here, to avoid underflow of computation, we manipulate $\log( p(y_{1:n}^{(i)} \vert x_{1:n}))$ instead of $p(y_{1:n}^{(i)} \vert x_{1:n})$.

\begin{algorithm}
	\caption{MAP Sequence Estimation Algorithm}
	\label{MAP.algorithm}
	\begin{tabbing}
		\enspace Initialization: draw $N$ i.i.d. samples $y_1^{(i)}$ from initial distribution $p(y_1)$, let $a_1^{(i)}=\log (p(y_1^{(i)}))$, \\
		\enspace where $i=1,\ldots,N$\\
		\enspace For $t=2$ to $t=n$ \\
		\qquad For $i=1$ to $i=N$ \\
		\qquad \qquad Draw $\bar{y}_t^{(i)}$ in $\Omega$ independently from $p(y_t \vert y_{1:t-1}^{(i)})=p(y_t\vert y_{t-1}^{(i)}) $ \\
		\qquad \qquad Set $\bar{y}_{1:t}^{(i)}=\{y_{1:t-1}^{(i)},\bar{y}_t^{(i)}\}$ \\
		\qquad \qquad Evaluate the importance weights $\tilde{w}_t^{(i)}=p(x_t \vert \bar{y}_{1:t}^{(i)},x_{1:t-1})=p(x_t \vert x_{t-1},\bar{y}_t^{(i)})$ \\
		\qquad For $i = 1$ to $i = N$ \\
		\qquad \qquad Normalize the importance weights to obtain $w_t^{(i)}={\tilde{w}_t^{(i)}}\big/{\sum_{k=1}^N\tilde{w}_t^{(k)}}$\\
		\qquad For $i=1$ to $i=N$ \\
		\qquad \qquad Compute $\bar{a}_{t}^{(i)}=\bar{a}_{t-1}^{(i)}+\log(p(\bar{y}_{t}^{(i)}\mid y_{t-1}^{(i)}))+\log(p(x_t \mid x_{t-1}, \bar{y}_t^{(i)}))$ \\
		\qquad \qquad Set $y_{1:t}^{(i)}=\bar{y}_{1:t}^{(k)}$ and $a_t^{(i)} = \bar{a}_t^{(k)}$ with probability $w_t^{(k)}$, where $k\in \{1,\ldots,N\}$\\
		\enspace Set the maximizer $\hat{y}_{1:n}^N=y_{1:n}^{N_0}$, where $N_0=\argmax_{i\in \{1,\ldots,N\}} a_n^{(i)}$
	\end{tabbing}
	\vspace*{-10pt}    
\end{algorithm}

The following theorem establishes the consistency of Algorithm \ref{MAP.algorithm} for the TASS model. 

\begin{theorem}\label{coro.TASS.MAP}
	Let $\hat{y}_{1:n}^N$ be the output sequence from \eqref{max.appro} and $\hat{y}_{1:n}$ be the solution defined in \eqref{maximization}. If the data follows the TASS model, then, almost surely,
	\begin{eqnarray*}
		\lim_{N \to \infty}p(\hat{y}_{1:n}^N \vert x_{1:n})=\max_{y_{1:n}\in\Omega^n} \left\{p(y_{1:n} \vert x_{1:n})\right\}\,.
	\end{eqnarray*}
\end{theorem}
By Theorem \ref{coro.TASS.MAP}, $\hat{y}_{1:n}^N$ can serve as an estimate of the latent process $\{Y_t\}_{t=1,\ldots,n}$. Therefore, change-points can be estimated as time points at which the estimated latent process hits the estimated threshold.
With the estimated latent process $\hat{y}_{1:n}^N$, denote the threshold to be hit at the first change-point as $\hat{r}_i = \inf\{\hat{r}_j: \hat{r}_j > \hat{y}_1^N, j =1,\ldots,\hat{m}\}$. Thus, the change-points can be estimated as
\begin{eqnarray}
	\hat{t}_1&=&\min\{t> 0: \hat{y}^{N}_t \geq \hat{r}_i \}\,, \  
	\hat{t}_2 =\min\{t> \hat{t}_1: \hat{y}^{N}_t \geq \hat{r}_{{ (i + 1) \mbox{ mod } \hat{m}}} \}\,,  \ldots\,,  \nonumber \\  
	\hat{t}_k&=&\min\{t> \hat{t}_{k-1}: \hat{y}^{N}_t \geq \hat{r}_{ {(i+ k-1) \mbox{ mod } \hat{m}}} \}\,,\ldots \label{eq:est.latent}
\end{eqnarray}

\subsection{Diagnostic Checks} \label{sec:diag_check}

In general, predicting future change-points is a challenging problem. Clearly, when there are too few change-points in the data, or the occurrence of change-points has no recurring pattern or structure, there is little hope to have a good prediction. The proposed TASS model provides one particular structure on the occurrence of change-point that is helpful for predicting future change-points. Although this model may not be applicable to all data sets, we develop a 
diagnostic check procedure to determine whether a given time series is suitable to be analysed by the TASS model. 

First, we define the residuals of the TASS model. Suppose that we have fitted an $m$-regime
TASS model based on \eqref{model} and \eqref{Y.RW}, where $\epsilon_{t}\overset{i.i.d.}{\sim}\text{Gamma}(\alpha,\beta)$
and $e_{t}\overset{i.i.d.}{\sim}\text{N}(0,1)$, to a time series $\{ X_{1},\ldots,X_{n}\}$.
Denote the estimated parameters as $\hat{\theta}_{m}=(\hat{\phi}_{1},\hat{a}_{1},\hat{\sigma}_{1},...,\text{\ensuremath{\hat{\phi}_{m},\hat{a}_{m},\hat{\sigma}_{m}},\ensuremath{\hat{r}_{1},...,\text{\ensuremath{\hat{r}_{m-1},\hat{\alpha},\hat{\beta}}}}})$, 
and the estimated latent process as $\hat{Y}_{t}$. 
Based on \eqref{model} and \eqref{Y.RW}, we define the residuals as 

\begin{eqnarray}
	\hat{e}_{t} & =&\begin{cases}
		\dfrac{X_{t}-\hat{a}_{1}-\hat{\phi}_{1}(X_{t-1}-\hat{a}_{1})}{\hat{\sigma}_{1}}\,,  & \hat{Y}_{t}\in[0,\hat{r}_{1})\,,\\
		\dfrac{X_{t}-\hat{a}_{2}-\hat{\phi}_{2}(X_{t-1}-\hat{a}_{2})}{\hat{\sigma}_{2}}\,, & \hat{Y}_{t}\in[\hat{r}_{1},\hat{r}_{2})\,,\\
		\vdots\\
		\dfrac{X_{t}-\hat{a}_{m}-\hat{\phi}_{m}(X_{t-1}-\hat{a}_{m})}{\hat{\sigma}_{m}}\,, & \hat{Y}_{t}\in[\hat{r}_{m-1},1)\,,
	\end{cases} \label{eq: esti.XRes} \\
	\hat{\epsilon}_{t} & =&\begin{cases}
		\hat{Y}_{t}-\hat{Y}_{t-1}\,, & \hat{Y}_{t-1}\leq\hat{Y}_{t}\,,\\
		\hat{Y}_{t}-\hat{Y}_{t-1}+1\,, & \hat{Y}_{t-1}>\hat{Y}_{t}\,.
	\end{cases} \label{eq: esti.YRes}
\end{eqnarray}

Discrepancy between the data and estimated model can be measured by investigating the residuals $\{\hat{e}_t\}$ and $\{\hat{\epsilon}_t\}$. If the TASS model fits the data adequately, then $\{\hat{e}_t\}$, $\{\hat{\epsilon}_t\}$ should be white noise with marginal distributions following $\text{N}(0,1)$ and Gamma$(\hat{\alpha},\hat{\beta)}$, respectively. These properties can be readily checked by applying standard procedures such as Ljung-Box test, Anderson–Darling test, QQ-plot and ACF plot. See Section \ref{sec:real} below for illustrations. 

\section{Prediction of Future Change-points} \label{prediction.method}

In this section, we discuss the prediction of future change-points using the TASS model. Recall that $X_t$ enters the next regime at the time when the latent process $Y_t$ crosses a threshold. Therefore, predicting the next change-point is equivalent to predicting when $Y_t$ hits its succeeding threshold.

Denote $r_k(y_n)$ as the $k$th subsequent threshold value given the latent process is at $y_n$, i.e. $r_k(y_n)= \inf\{r_{(j+k-1)\mod{m{}}} \mid r_j>y_n, j=1,\ldots,m\}$. 
The sample version $\hat{r}_{k}(y_n)$ can be analogously defined.  
Since the range of the 
threshold values is $[0,1]$ and 
the latent process repeatedly travels across the range for $\left\lfloor \frac{k-1}{m} \right\rfloor $ times
before reaching the $k$th subsequent threshold value, 
the time till the $k$th change-point $T_k$ can be represented as 
\begin{eqnarray*}
	T_k=\mbox{inf}\left\{t \in\mathbb{Z}^{+}: \sum_{j=n+1}^{n+t} \epsilon_j \geqslant r_k(y_n)+\left\lfloor \frac{{k-1}}{m} \right\rfloor -y_n, \  t=1,2,\ldots \right\} \,.
\end{eqnarray*}

As the current time is $n$, the time of the next change-point $\tau_k$ is given by $\tau_k=n+T_k$. Our goal is to characterize the distribution of $T_k$ given the data $x_{1:n}$. 

Observe that the events $\{T_k>t\}$ and $\left\{\sum_{j=n+1}^{n+t}\epsilon_j <r_k(y_n) +\left\lfloor \frac{k-1}{m} \right\rfloor -y_n\right\}$ are equivalent. 
Thus,
\begin{eqnarray}
	&& P\left(T_k>t | x_{1:n}, \theta_m \right) \nonumber \\
	&=&\int P\left(T_k>t | y_{1:n}, x_{1:n}, \theta_m)P(y_{1:n}| x_{1:n}, \theta_m \right)dy_{1:n} \nonumber \\ 
	&=& \int P\left(\sum_{j=n+1}^{n+t}\epsilon_j < r_k(y_n)-y_n \right) P\left(y_{1:n}| x_{1:n},  \theta_m\right)dy_{1:n}\,. \label{pred.int}
\end{eqnarray}

We now discuss a simple procedure to estimate $P(T_k>t | x_{1:n}, \theta_m)$ in \eqref{pred.int}. 
First, based on the model setting, $\sum_{j=n+1}^{n+t}\epsilon_j \sim \mbox{Gamma}(t \alpha, \beta)$. Thus, 
\begin{eqnarray*}
	P\left(\sum_{j=n+1}^{n+t}\epsilon_j < r_k(y_n)-y_n \right) = G(r_k(y_n)-y_n ; t\alpha,\beta)\,,
\end{eqnarray*}
where $G(\cdot;a,b)$ is the distribution function of $\mbox{Gamma}(a,b)$.  
Next, the integral in \eqref{pred.int} can be approximated by Monte Carlo. 
Specifically, if we have conducted the MAP sequence estimation, we have  
simulated $N$ paths of $y_{1:n}^{(i)}, i=1,\ldots,N$ by Algorithm \ref{MAP.algorithm} based on the posterior distribution $p(y_{1:n}|x_{1:n},\hat{\theta}_m)$. Therefore, we can approximate $P(T_k> t \mid x_{1:n},\theta_m)$ by
\begin{align} \label{CP.dist}
	\hat{P}\left(T_k> t \mid x_{1:n},\hat{\theta}_m\right) = \frac{1}{N}\sum_{i=1}^{N}G\left( \hat{r}_k(y_n^{(i)}) +\left\lfloor \frac{k-1}{m} \right\rfloor -y_n^{(i)}; t\hat{\alpha}, \hat{\beta}\right)\,.
\end{align}

With the distribution of $T_k$ in \eqref{CP.dist}, we can predict the next change-point and further generate its prediction interval. 
Specifically, the predictor of the next change-point can be constructed as 
\begin{eqnarray}\label{predict1}
	\hat{\tau}_k=n+\widehat{\mathbb{E}}(T_k) = n+\sum_{t=1}^{\infty}\sum_{i=1}^{N}\frac{1}{N} G\left(\hat{r}_k(y_n^{(i)}) +\left\lfloor \frac{k-1}{m} \right\rfloor - y_n^{(i)}; t\hat{\alpha}, \hat{\beta}\right)\,.
\end{eqnarray} 

Furthermore, the $(1-\alpha_0)\%$  prediction interval $(\hat{\tau}_k^{(l)}, \hat{\tau}_k^{(r)})$ can be obtained by solving the equations
\begin{align*}
	\frac{1}{N}\sum_{i=1}^{N}\left[1-G\left(\hat{r}_1(y_n^{(i)}) +\left\lfloor \frac{k-1}{m} \right\rfloor -y_n^{(i)};\hat{\tau}_1^{(l)}\hat{\alpha}, \hat{\beta}\right)\right]&=\frac{\alpha_0}{2}\,,\\
	\frac{1}{N}\sum_{i=1}^{N}\left[1-G\left(\hat{r}_1(y_n^{(i)}) +\left\lfloor \frac{k-1}{m} \right\rfloor -y_n^{(i)};\hat{\tau}_1^{(r)}\hat{\alpha}, \hat{\beta}\right)\right]&= 1-\frac{\alpha_0}{2} \,.
\end{align*}

\section{Simulation Studies} \label{sec:sim}

\subsection{Parameter Estimation}
The first time series $\{X_t\}$ is generated from the following two-regime TASS model:
\begin{align} \label{TASS_gam_50_e}
	\begin{split}
		X_t &=
		\begin{cases}
			-3-0.3(X_{t-1}+3)+ e_t\,, &Y_t \in [0,0.6)\,, \\
			2+0.6(X_{t-1}-2)+2 e_t\,, &Y_t \in [0.6,1) \,,
		\end{cases} \\ 
		Y_t &=
		\begin{cases}
			Y_{t-1}+ \epsilon_t\,, &\mbox{if } Y_{t-1} + \varepsilon_t <1\, ,  \\
			Y_{t-1}+ \epsilon_t- \lfloor Y_{t-1}+ \epsilon_t \rfloor\,,  &\mbox{otherwise}\, ,
		\end{cases}
	\end{split}
\end{align}
where $\epsilon_t \overset{i.i.d.}{\sim} \mbox{Gamma}(0.5, 50)$ and $e_t \overset{i.i.d.}{\sim} \mbox{N}(0,1)$. The average length of the first and the second regime are 60 and 40, respectively. 

\begin{table}
	\caption{Estimation results for the TASS(2) model (21).}
	\label{gam.esti}
	\centering
	\begin{tabular}{ccccccccccc}
		\hline
		&Parameter  &$\phi_1$  & $\phi_2$ & $a_1$  & $a_2$  & $\alpha$  & $\beta$ & $r_1$ &$\sigma_1$&$\sigma_2$ \\\hline
		&True value & $-0.3$& $0.6$  & $-3$  & $2$  & $0.5$   & $50$  & $0.6$  &1&2\\
		$n=1000$&Mean & -0.303 &  0.593 & -3.003&  1.885 &  0.611& 58.84 &
		0.602 &  1.005 &  1.994 \\
		&RMSE & 0.035 &  0.046 &  0.032 &  0.297 &  0.273 &
		23.21 & 0.026 &  0.032&  0.075\\
		$n = 2000$ & Mean &  -0.304 &  0.599 & -3.004 &  1.880 &  0.563 & 54.02 & 0.602 &  1.005 &  2.002 \\
		& RMSE & 0.024 &  0.033 &  0.022 &  0.228 &  0.166 &
		14.74 & 0.017 &  0.022 &  0.052 \\
		$n=3000$ & Mean &   -0.304 &  0.601 & -3.003&  1.885 &   0.517 & 49.35 & 0.601 & 1.005 &  2.006\\
		& RMSE & 0.019 &  0.026 &  0.019 &  0.196 &  0.124 & 10.08&
		0.014 &  0.019 &  0.043\\
		\hline
	\end{tabular}
\end{table}

We estimate the parameter vector $\theta_2=(\phi_1, \phi_2,a_1, a_2, \alpha, \beta, r_1 , \sigma_1, \sigma_2)$ under sample sizes of $n = 1000, 2000, 3000$, respectively. 
The number of replications for each setting is 500. 
Note that the second term $\sum_{j=1}^{\infty}g_{\alpha,\beta}(y_t-y_{t-1}+j)$ on the right hand side of \eqref{conditionpr2} needs to be truncated by some constant $C$, 
i.e, $\sum_{j=1}^{C}g_{\alpha,\beta}(y_t-y_{t-1}+j)$. 
We suggest $C=100$ such that the error caused by the truncation is negligible. 
The estimates and root mean square errors (RMSE)  are reported in Table \ref{gam.esti}. 
The estimates of all parameters are close to the true value. 
Also, as the sample size increases, the RMSE of the parameter estimates decreases.

\subsection{Change-point Prediction}

Next, we investigate the performance of the prediction algorithm for models \eqref{TASS_gam_50_e}. 
First, we estimate the MAP sequence $\hat{y}_{1:n}$ by Algorithm $\ref{MAP.algorithm}$ with $N=500$. 
Then, we compute the predictor $\hat{\tau}$ for the next change point $\tau$ and construct prediction intervals using the procedure proposed in Section \ref{prediction.method}. 
To systematically study the prediction accuracy, we report the root-mean prediction error (PE) of 500 replications of the predictor $\hat{\tau}$ in Table \ref{gam.PE}. 
Moreover, we calculate the coverage rate (CR) of the prediction intervals. 

It can be seen that the prediction error decreases as the sample size increases. 
Also, the coverage rate is close to the confidence level of the prediction interval, 
which reflects the precision of the prediction intervals.

\begin{table}[!h]
	\centering
	\caption{Prediction errors and coverage rates of prediction intervals for the TASS(2) model \eqref{TASS_gam_50_e}.}
	\begin{tabular}{ccccc}
		\hline
		
		& Prediction Error & \multicolumn{3}{c}{Coverage rate} \\
		\cline{3-5}
		& & 80\% P.I.  & 90\% P.I. & 95\% P.I. \\
		\hline
		$n=1000$  &10.97 & 78.2\%& 86.8\%& 93.4\%\\
		$n=2000$ &  10.49 & 79.2\% & 89.6\% & 94.2\%\\
		$n=3000$ & 10.23 & 79.6\% & 89.8\% & 95.0\%\\
		
		\hline
	\end{tabular}
	\label{gam.PE}
\end{table}

\subsection{Diagnostic Checks}

In this section, we investigate the performance of the proposed diagnostic checks through studying the empirical sizes and powers of Ljung-Box test and Anderson-Darling test for the residuals $\{\hat{e}_t\}$ and $\{\hat{\epsilon}_t\}$ of model \eqref{TASS_gam_50_e}, respectively. 
Sample sizes of $n=1000, 2000, 3000$ are considered with 500 replications in each case. 
Table \ref{tb:empi_size} reports the empirical sizes for a nominal level of $\alpha = 5 \% $. 
Observe that the empirical sizes approach 5$\%$ as the sample size increases. 

\begin{table} [!h]	
	\caption{Empirical sizes for a nominal level of 5$\%$ in the Ljung-Box test (LBT) and  Anderson-Darling test (ADT) for the residuals $\{\hat{e}_t\}$ and $\{\hat{\epsilon}_t\}$ of model (21).} 
	\label{tb:empi_size}
	\centering
	\begin{tabular}{ccccc}
		\hline
		& LBT for $\{\hat{e}_t\}$ & LBT for $\{\hat{\epsilon}_t\}$ & ADT for $\{\hat{e}_t\}$ & ADT for $\{\hat{\epsilon}_t\}$\\ \hline
		$n = 1000$ & 6.4$\%$ & 3.8$\%$ & 7$\%$ & 4.2$\%$\\
		$n = 2000$ & 4.2$\%$ & 4.2$\%$ & 6$\%$ & 5.4$\%$\\
		$n = 3000$ & 4.6$\%$ & 4.8$\%$ & 5.2$\%$ & 4.4$\%$\\
		\hline
	\end{tabular}
\end{table}

Next, we study the empirical power of the diagnostic check for TASS$(2)$ model when the time series data is not suitable for TASS$(2)$. 

Let $I_1=[s_1:e_1] \coloneqq \{s_1,s_1+1,\cdots ,e_1 \}$, where $s_1 =1, e_1 = 12$ and $I_i = [e_{i-1}+1:e_{i-1}+10+2i]$ for $i = 2,3,...,m.$  We simulate $\{X_t\}$ from the model 

\begin{eqnarray} 
	\label{eq:model_power}
	X_t = \begin{cases}
		3+0.3(X_{t-1}-3)+e_t\,, & t \in I_j  \text{ and } j \equiv 1 (\text{mod 2})\,,\\
		0.3X_{t-1}+e_t\,, &t \in I_j  \text{ and } j \equiv 0  (\text{mod 2})\,,
	\end{cases}
\end{eqnarray}
where $e_t \overset{iid}{\sim} t_5$. Since the mean changes in model \eqref{eq:model_power} occur with increasing time lags, the residual $\{\hat{\epsilon}_t\}$ cannot be modeled by a gamma distribution with a pair of fixed ($\hat{\alpha}$, $\hat{\beta}$). 
Also, $\{\hat{e}_t\}$ is not normally distributed. 
We conduct the Ljung-Box test and Anderson-Darling test for the residuals $\{\hat{e}_t\}$ and $\{\hat{\epsilon}_t\}$ of model \eqref{eq:model_power} under $m = 30, 40, 50$, corresponding to sample sizes $n = 1230, 2040, 3050$, respectively. 
Table \ref{tb:empi_power} reports the empirical powers for a nominal size of 5$\%$. 
Observe that the empirical powers increase as $m$ increases. 
Also, the empirical powers of Anderson-Darling test to $\{\hat{e}_{t}\}$ and $\{\hat{\epsilon}_{t}\}$ are close to 1 when $m$ is large.  
Therefore, the diagnostic checks successfully detect the data that is not suitable to be modeled by 
a TASS model.

\begin{table}[!h]
	
	\caption{Empirical powers for a nominal level of 5$\%$ in the Ljung-Box test(LBT) and Anderson-Darling test (ADT) for the residuals of TASS(2) model fitting based on data from (23).}
	\centering
	\begin{tabular}{ccccc}
		\hline
		& LBT for $\{\hat{e}_t\}$ & LBT for $\{\hat{\epsilon}_t\}$ & ADT for $\{\hat{e}_t\}$ & ADT for $\{\hat{\epsilon}_t\}$\\
		$m = 30$ & 44.4\% & 5.6\% & 99.6\% & 85\% \\
		$m = 40$ & 51.2\% & 7.8\% & 100\% & 96\% \\
		$m = 50$ & 61.4\% & 10.2\% & 100\% & 98.6\% \\
		\hline
		
	\end{tabular}
	\label{tb:empi_power}
\end{table}

\section{Real Data Application}\label{sec:real}

In this section we apply the proposed methodology to the energy consumption data from PJM Interconnection LLC, which is a regional transmission organization  in the United States. 
This dataset consists of hourly energy consumption in Washington, DC covering the period  2005-2018, and is available on the website \textit{https://www.kaggle.com/robikscube/hourly-energy-consumption}. 
To ease computation, we convert the hourly data into 630 observations of weekly energy consumption. 
The time series is plotted in Figure \ref{week.pic}. 
It can be observed that the energy consumption volume admits different behaviors under different periods. 
Predicting future change-points in the data can potentially help the company better allocate resources to meet the electricity consumption needs.  

\begin{figure}[h!]
	\centering
	\caption{Weekly energy consumption data. The units of $x$-axis and $y$-axis are week and $10^3$ Megawatts, respectively.}
    \includegraphics[width = \textwidth]{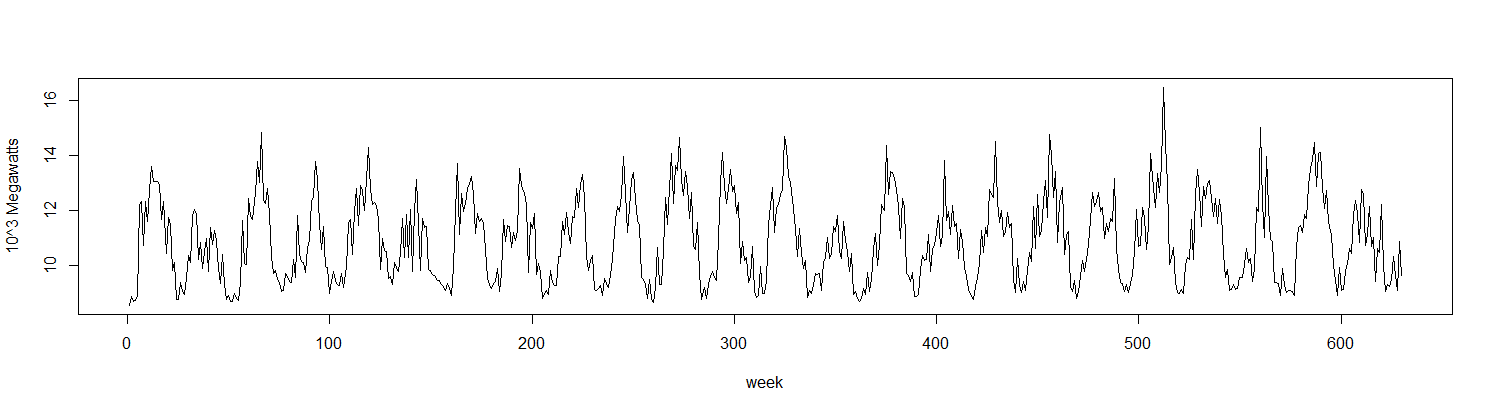}
	
	\label{week.pic}
\end{figure} 

To evaluate the effectiveness of the proposed TASS model, we split the dataset into training data~($t = 1, \ldots, 555$) and test data~($t = 556, \ldots, 630$). For the training data, we compare the BIC of one, two and three-regime TASS models, yielding two as the estimated number of regimes.

The CTL$_2$ estimation is applied to both the training and full data, and the results are given in Table \ref{est.par}. 
Observe that the two groups of estimates are quite close to each other. 
Also, the intercept and standard deviation of the two regimes have large differences. The intercept $a$ of the two regimes are around 9 and 11, respectively. 
The standard deviation $\sigma$ of the two regimes are 0.3 and 1, respectively. 
This shows the difference in  electricity consumption demand under different regimes. During the second regime, the electricity consumption is higher and fluctuates more.   The threshold value is around 0.3, which suggests that the average time of the high electricity consumption regime is around 70\%.  

We estimate change-points by the MAP sequence estimation method proposed in Section \ref{sec:MAP}. 
Figure \ref{one-year.pic} plots the estimated change-points in the first 60 weeks based on the training data and full data. 
It can be seen that there are roughly two peaks and troughs of electricity consumption within one year. The dataset starts from May 2005. 
Matching the change-points with the real-time periods, we found there are two high electricity consumption periods. 
One is from middle June to middle October. 
The other is from the end of November to early April of the next year. 
That is, the electricity consumption demand is high in summer and winter. 
Similarly, the remaining period in Figure \ref{one-year.pic}, spring and autumn, corresponds to the low electricity consumption demand regime. 
The remaining change-points not shown in Figure \ref{one-year.pic} also exhibit the same pattern.

\begin{figure}[h!]
	\centering
	\caption{Change-points estimated by MAP algorithm in first 60 weeks. Left panel: results from training data. Right panel: results from full data.}
	\includegraphics[width=\textwidth]{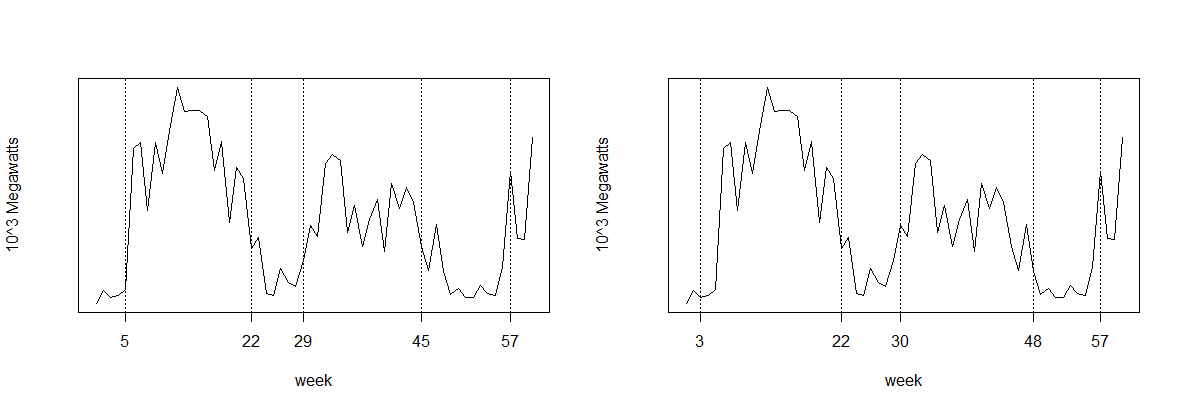}
	
	\label{one-year.pic}
\end{figure} 

\begin{table}
	\centering
	\caption{Parameter estimation results for the PJM dataset using a two-regime TASS model.)} 
	\begin{tabular}{cccccccccc}
		\hline
		Parameter  &$\phi_1$  & $\phi_2$ & $a_1$  & $a_2$  & $\alpha$ & $\beta$ & $r_1$ &$\sigma_1$&$\sigma_2$ \\\hline 
		Estimates (full data)  &0.328 &  0.617 &  9.276&  11.60 &  0.223&  4.910& 0.279 &  0.350 &  1.045
		\\ 
		Estimates (training data) & 0.324 &  0.628 &  9.272 & 11.60 &  0.211 &  5.088 &  0.298 &  0.351 &  0.990\\
		\hline
	\end{tabular}
\label{est.par}
\end{table}

Next, based on the training data, we predict the future change-points after time $t = 555$ using the prediction method proposed in Section \ref{prediction.method}. 
The prediction results are reported in Table \ref{pred.result}. 
For comparison, we also estimate the change-points using the full data, resulting in six change-points after $t=555$, which are treated as the ``true" value of the change-points. 
Denote $\tau_i$ as the $i$-th change-point in the test data. It can be seen from Table \ref{pred.result} that  
the predicted values of $\tau_1, \ldots, \tau_6$ are very close to  the ``true" values. 

\begin{table}[!h]
	\centering
	\caption{Change-point prediction results for the PJM dataset.} 
	\begin{tabular}{cccccccc}
		\hline
		&``True" value&  Predicted value & 80\% P.I. && 90\% P.I. && 95\% P.I. \\ \hline 
		$\tau_{1}$  & 567 & 569 & (560, 580) && (558, 584) && (557, 588) \\ 
		$\tau_{2}$  & 578 & 577 & (565, 590) && (562, 595) && (560, 599)\\
		$\tau_{3}$  & 596 & 594 & (577, 612) && (574, 618) && (570, 623)\\
		$\tau_{4}$  & 603 & 602 & (584, 621) && (579, 627) && (576, 633)\\  
		$\tau_{5}$  & 621 & 619 & (597, 641) && (592, 648) && (587, 655)\\
		$\tau_{6}$  & 629 & 626 & (604, 650) && (598, 657) && (593, 664)\\
		\hline
	\end{tabular}
	
	\label{pred.result}
\end{table}

To evaluate the overall performance of change-point detection and prediction, we plot all estimated and predicted change-points in Figure \ref{pred.pic}. 
The number of change-points detected by the full data and training data before $t=555$ both equal to 43. 
Together with the six change-points to be predicted in the test data, there are in total 49 change-points in the dataset. 

The change-points detected based on the full data and training data are depicted by bottom dots and top dots, respectively. 
Also, the predicted change-points based on the training data are plotted by triangles. 
From Figure \ref{pred.pic}, we can observe that the change-points detected and predicted based on the training data are not far from the change-points detected based on the full data.

To investigate the overall accuracy of change-points predicted in test data, we calculate the prediction error of change-points. 
Denote $\hat{\tau}_i$ and $\tau_i$ as the predicted value and ``true" value of the $i$th change-point, respectively.  
The rooted mean squared prediction error is $\sqrt{ \left.\sum_{i=1}^{6}(\hat{\tau_i}-\tau_i)^2\right/6  } = 1.96$, indicating that the  change-points can be predicted within 1.96 weeks on average. 

\begin{figure}[!h]
	\centering
	\caption{Change-points detected based on the full data (bottom dots) and training data (top dots); change-points predicted based on the training data (top triangles).}
	\includegraphics[width=\textwidth]{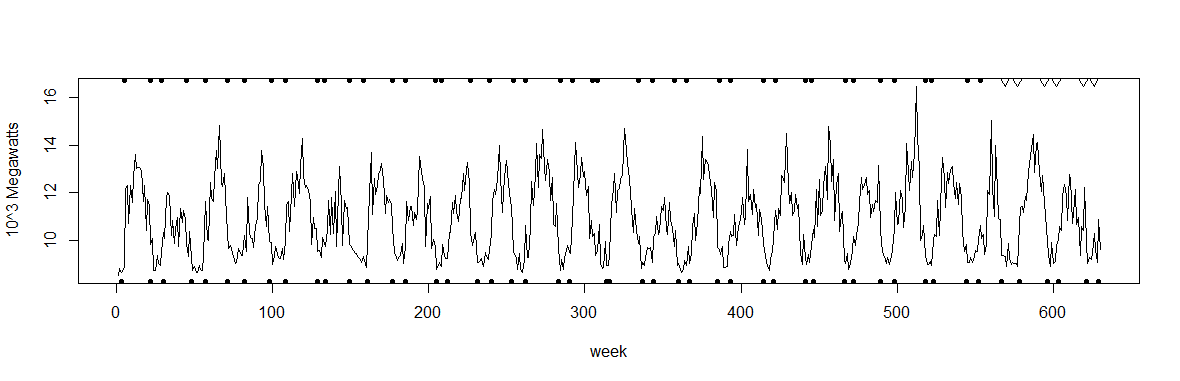}
	
	\label{pred.pic}
\end{figure}

In time series analysis, a sequence with periodic behavior can also be analyzed by models with seasonal effects, for example, the SARIMA model. 
We fit an SARIMA model based on the training data to predict the future observations. 
We try different orders of SARIMA and use BIC to conduct model selection, which results in the following SARIMA$(1,0,0)\times(1,0,0)_{27}$ model:
\begin{eqnarray*}
	X_t = 2.4 + 0.72X_{t-1}+0.23X_{t-27}-0.17X_{t-28}+e_t , \quad e_t \sim N(0, 0.8552)\,.
\end{eqnarray*}
For comparison, the prediction is also conducted by the TASS model based on the training data. Specifically, given the simulated $N$ paths of $y_{1:n}^{(i)}, i = 1, \ldots,N$ by Algorithm 2, we can further simulate $N$ paths of future $y_t^{(i)}$ and $x_t^{(i)}$, $t >n, i=1,\ldots,N$ using the parameter estimates. 
Hence, the future $X_t$ is predicted as $1/N\sum_{i=1}^N x_t^{(i)}$.

The prediction accuracy of the two methods are compared on the test data. 
Out of the 75 weeks electricity volumes to be predicted, the prediction of TASS model outperforms the SARIMA model for 52 weeks. 
Denote $X_t$ and $\hat{X}_t$ as the true value and predicted value of the electricity volume. 
The root mean square error (RMSE), mean absolute error (MAE) and mean absolute percentage error (MAPE) of the prediction on test data are calculated by 
$ \sqrt{\left.  \sum_{t = 556}^{630}(\hat{X}_t-X_t)^2\right/ 75} $, 
$\left.\sum_{t=556}^{630} \left| \hat{X}_t -X_t \right| \right/ 75$ and
$\left.\sum_{t=556}^{630} \left| {\hat{X}_t -X_t}\right| \right/(75{X_t})$, respectively. 
The results are reported in Table \ref{pred.comparison}.
The RMSE, MAE and MAPE of the TASS model are all smaller than those of SARIMA, indicating that the prediction power of the TASS model is better than the classical time series model with seasonality.

\begin{table}[!h]
	\centering
	\caption{Comparison of the prediction result based on the TASS model and SARIMA model.} 
	\begin{tabular}{cccc}
		\hline
		& RMSE & MAE & MAPE \\ \hline 
		SARIMA & 1.52 & 1.27 & 11.55\% \\
		TASS & 1.17 & 0.88 & 7.78\% \\
		\hline
	\end{tabular}
	
	\label{pred.comparison}
\end{table}

Finally, we conduct a diagnostic check for the real data. We use the training data set to conduct the diagnostic check procedure discussed in Section \ref{prediction.method}. 
For any $n$ i.i.d. random variables $\{X_i\}_{i=1}^{n}$, the order statistic $X_{(k)}$ follows a $ \text{Beta}(k,n+1-k)$ distribution. The $(1-\alpha) $ confident band for $X_{(k)}$ can be approximated by $[F^{-1}(\beta_{\alpha /2,k}), F^{-1}(\beta_{1-\alpha/2, k})]$, where $F^{-1}(\cdot)$ is the quantile function of distribution of $X$ and $\beta_{\tau,k}$ is the $\tau$-quantile of Beta$(k,n+1-k)$ distribution.
Figure \eqref{qq_normal} provides the QQ-plot between N$(0,1)$
and $\{\hat{e}_{t}\}$ with dotted lines representing $95\%$ confident bands.  
Similarly, Figure \eqref{qq_gamma}  provides the QQ-plot between Gamma$(\hat{\alpha},\hat{\beta})$
and $\{\hat{\epsilon}_{t}\}$. 
It can be seen that the majority of points fall along the  QQ lines and within the corresponding 95$\%$ confidence bands. 
That is, the quantile of the residuals $\{\hat{e}_t\}$ and $\{\hat{\epsilon}_t\}$ agree with that of N$(0,1)$ and Gamma$(\hat{\alpha},\hat{\beta})$ distribution, respectively. 
We use the methods discussed in Section \ref{sec:diag_check} to find the residuals $\{\hat{e}_{t}\}$ and $\{\hat{\epsilon}_{t}\}$ of the data, and apply the Anderson-Darling test to each of the residuals. We obtained a $p$-value 0.639 and 0.215 for $\{\hat{e}_{t}\}$ and $\{\hat{\epsilon}_{t}\}$, respectively, which supports that the residuals fit the proposed marginal distributions. 
Figures \eqref{acf_normal} and \eqref{acf_gamma} provide the ACF plots of $\{\hat{e}_t\}$ and $\{\hat{\epsilon}_t\}$, respectively. 
These ACF plots suggest that $\{\hat{e}_t\}$ and $\{\hat{\epsilon}_t\}$ are indeed white noises. 
Also, we apply the Ljung-Box test with a lag equals 12 for autocorrelations in the residuals $\{\hat{e}_t\}$ and $\{\hat{\varepsilon}_t\}$, and the resulting $p$-values are 0.058 and 0.983, respectively. 
Thus, the autocorrelations of $\{\hat{e}_t\}$ and $\{\hat{\epsilon}_t\}$ are not significantly different from $0$. 
Similar results can be obtained when full data are used. In conclusion, the TASS model fits the data adequately. 

\begin{figure}[htp] 
	\centering
	\caption{QQ-plots and ACF plots of the residuals $\{{\varepsilon}_t\}$ and $\{{e_t}\}$ of training data.}
	\subfloat[QQ-plot between N$(0,1)$ and $\{\hat{e}_{t}\}$]{%
		\includegraphics[scale=0.22]{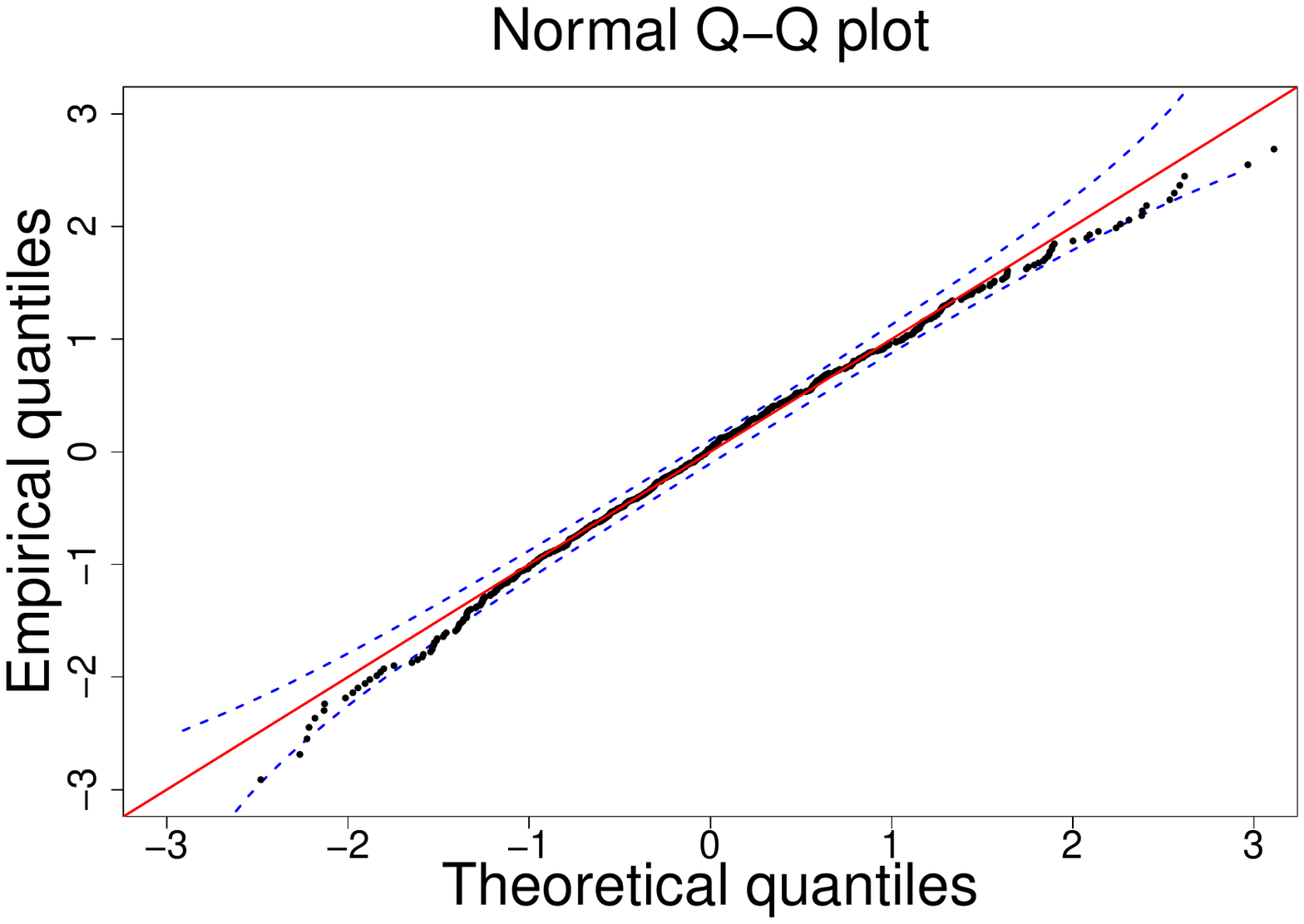} %
		\label{qq_normal}
	}%
	\subfloat[QQ-plot between  Gamma$(\hat{\alpha},\hat{\beta})$ and $\{\hat{\epsilon}_t\}$]{%
		\includegraphics[scale=0.22]{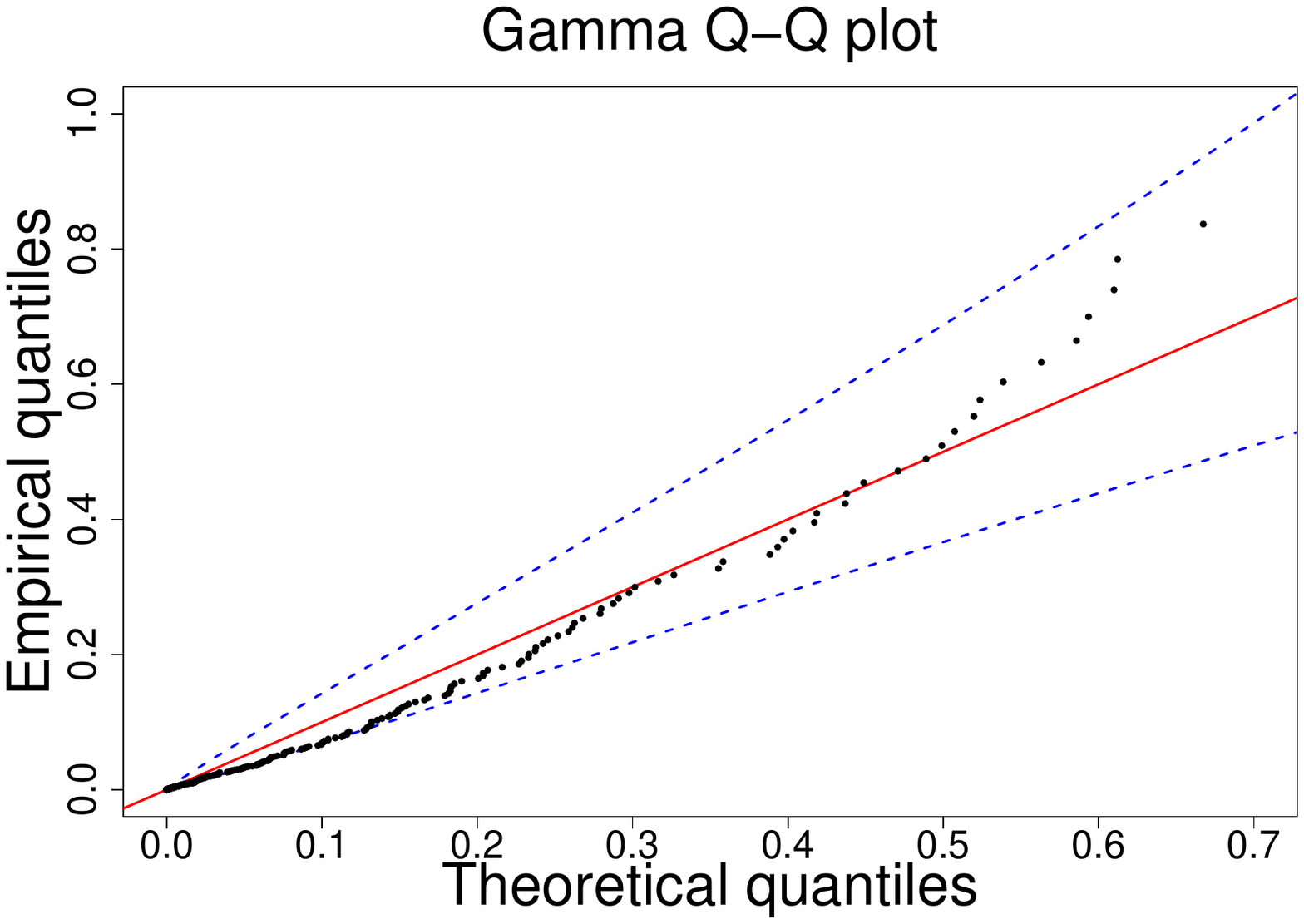}%
		\label{qq_gamma}
	}%
	\vskip\baselineskip
	\subfloat[ACF plot of $\{\hat{{e}}_t\}$]{
		\includegraphics[scale = 0.22]{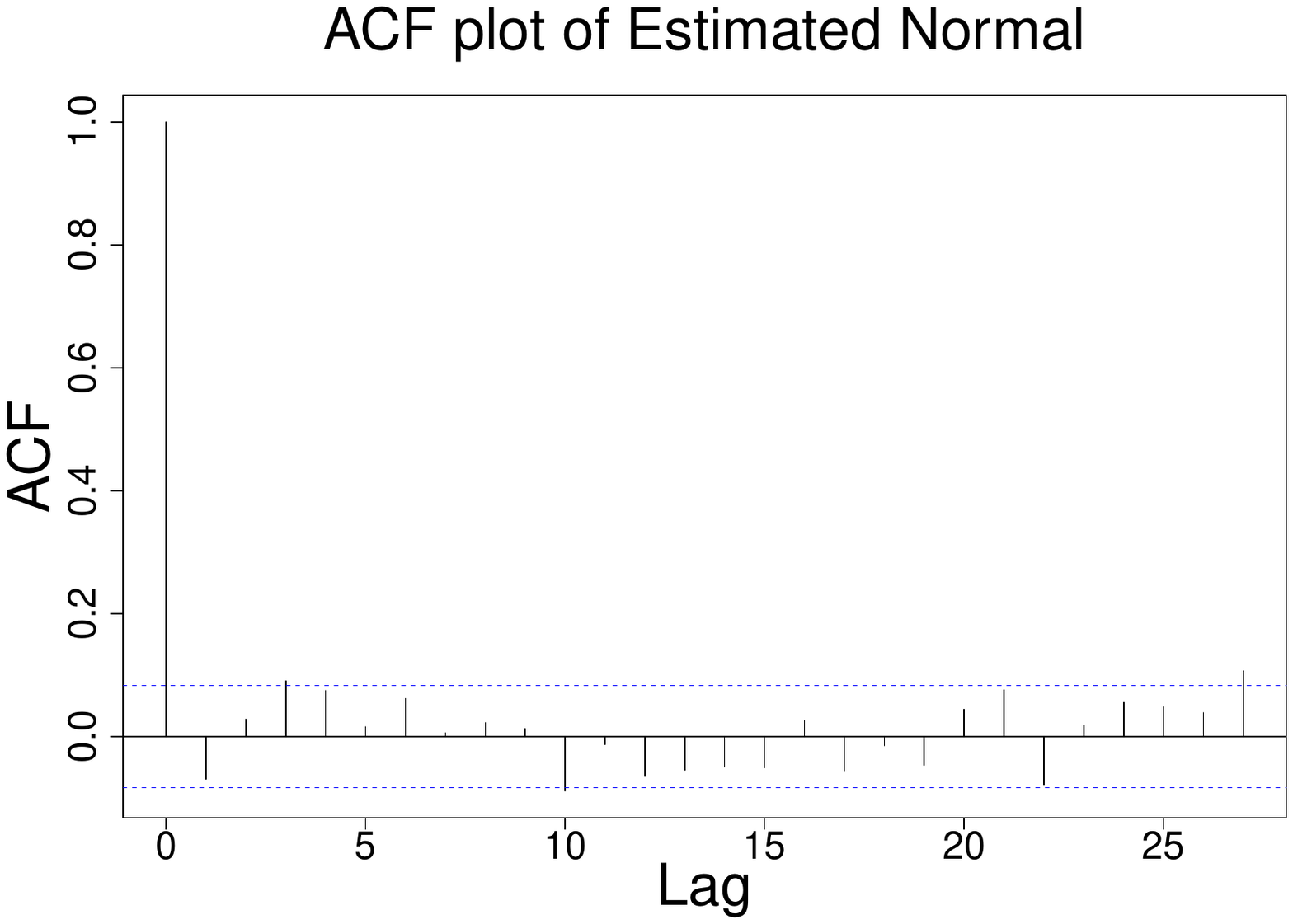}
	\label{acf_normal}
	}
	\subfloat[ACF plot of $\{\hat{\epsilon}_t\}$]{
		\includegraphics[scale = 0.22]{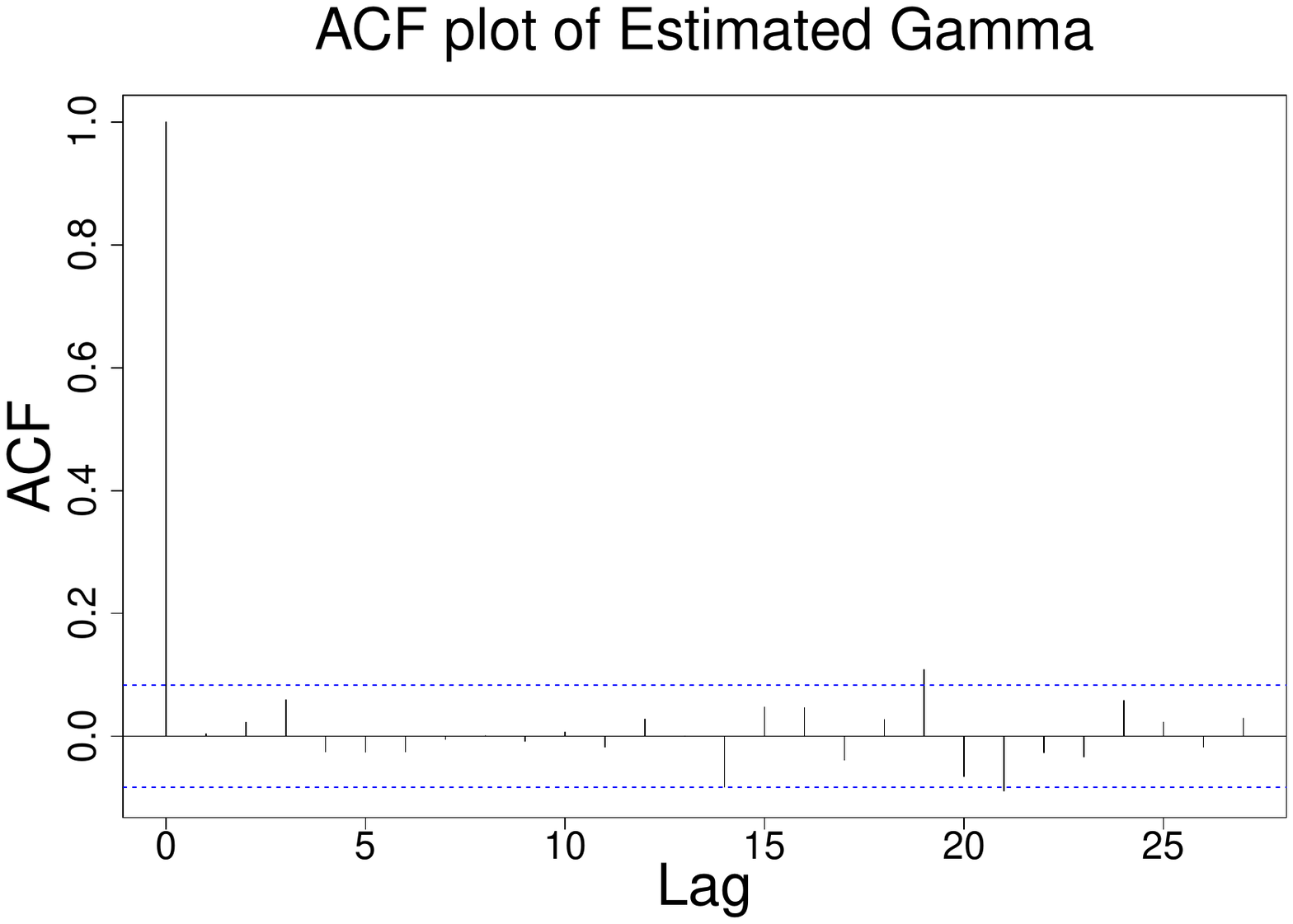}
		\label{acf_gamma}
	}

\end{figure}

\section*{Acknowledgments}
Research supported in part by grants HKRGC-GRF Nos 14302423/14302719/14304221.

\newpage

\appendix

\setcounter{equation}{0} 
\setcounter{lemma}{0} 
\renewcommand{\theequation}{\Alph{section}.\arabic{equation}}
\renewcommand{\thelemma}{\Alph{section}.\arabic{lemma}}

\section{Appendix}

{\bf Proof of Theorem \ref{stationary}.} For any positive continuous random variable $\epsilon_t$ with density $g(\cdot)$, it suffices to prove that for any Borel set $B$, 
\begin{eqnarray} \label{pi.staionary.general}
	\pi(B)=\int_{0}^{1}K(x,B)\pi(dx)\,,
\end{eqnarray}
where $K(x,\cdot)$ is the transition kernel defined by $K(x,B)=\int_B K(x,y)dy $, and $K(x,y)$ is the transition probability from $x$ to $y$.
Without loss of generality, let $B=(a,b)\,, 0<a<b<1. $ Then, from  \eqref{conditionpr2},
the transition kernel can be expressed as
\begin{eqnarray}
	K(x,B)= \int_{a \lor x}^b g(y-x) dy +
	\sum_{T=1}^\infty \int_a^b g(y+T-x) dy\,.
\end{eqnarray}
Thus, the right hand side of \eqref{pi.staionary.general} can be decomposed into three components:

\begin{align} \label{kernel.general}
	\int_{0}^{1}K(x,B)\pi(dx)
	&=\int_0^a \int_a^b g(y-x) dy dx + \int_a^b \int_x^b g(y-x) dy dx  \nonumber \\ 
	&\ \ \ \ + \sum_{T=1}^\infty \int_0^1 \int_a^b g(y+T-x) dydx\,.
\end{align}

Denote $G(\cdot)$ as the cumulative distribution function of $\epsilon_t$,
$\tilde{G} (a, b)$ as $G(b)-G(a)$,
each component can be computed as follows. The first component of \eqref{kernel.general} can be expressed as 
\begin{align} \label{component1.general}
	&\int_0^a \int_a^b g(y-x) dy dx 
	=\int_0^a \int_{a-x}^{b-x} g(y') dy' dx
	=\int_{0}^{b}\int_{0 \lor (a-y')}^{a \land (b-y') }  g(y') dx dy'\nonumber\\
	=&\int_{0}^{a\land (b-a)}\int_{a-y'}^{a }  g(y') dx dy'+
	\int_{b-a}^{a}\int_{a-y'}^{b-y' }  g(y') dx dy'\mathbbm{1}(b-a<a) \nonumber \\
	&+\int_{a}^{b-a}\int_{0}^{a} g(y') dx dy'\mathbbm{1}(b-a>a)+
	\int_{a \lor (b-a)}^{b}\int_{0}^{b-y' } g(y') dx dy' \nonumber\\
	=&\int_{0}^{a\land (b-a)}y'g(y') dx dy'+
	(b-a)\tilde{G}(a, b-a)\mathbbm{1}(b-a<a) 
	\nonumber \\&
	+ a\tilde{G}(b-a, a)\mathbbm{1}(b-a>a)+b \tilde{G}(b, a \lor (b-a))-
	\int_{a \lor (b-a)}^{b}y'g(y') dy'\,.
\end{align}
Next, for the second component of \eqref{kernel.general}, we have  
\begin{align}\label{component2.general}
	&\int_a^b \int_x^b g(y-x)dy dx =
	\int_a^b \int_0^{b-x} g(y')dy' dx  \nonumber\\=&
	\int_0^{b-a} \int_a^{b-y'} g(y') dxdy'=
	(b-a)G(b-a)-\int_0^{b-a} y'g(y')dy'\,.
\end{align}
Finally, the third component of \eqref{kernel.general} can be simplified as 
\begin{align}\label{component3.general}
	&\int_0^1 \int_a^b g(y+T-x) dy
	=\int_0^1 \int_{a-x+T}^{b-x+T} g(y') dy'dx \nonumber \\
	=&\int_{a+T-1}^{b+T} \int_{0 \lor (a+T-y')}^{1 \land (b+T-y')} g(y')dy'dx \nonumber \\
	=&\int_{a+T-1}^{b+T-1} \int_{a+T-y'}^{1} g(y')dxdy'
	+\int_{b+T-1}^{a+T} \int_{a+T-y'}^{b+T-y'} g(y')dxdy' \nonumber \\
	&+\int_{a+T}^{b+T} \int_{0}^{b+T-y'} g(y') dxdy' \nonumber\\
	=&(1-a-T)\tilde{G}(b+T-1, a+T-1)+\int_{a+T-1}^{b+T-1} y'g(y')dy' \nonumber\\
	&+(b-a)\tilde{G}(a+T, b+T-1)+(b+T)\tilde{G}(b+T, a+T) \nonumber \\
	&-\int_{a+T}^{b+T}y'g(y')dy'\,.
\end{align}
Substitute \eqref{component1.general}, \eqref{component2.general} and \eqref{component3.general} into \eqref{kernel.general}, we have
\begin{align}
	&\int_{0}^{1}K(x,B)\pi(dx) \nonumber\\
	=&bG(b)-aG(a)-\int_a^b y'g(y')dy'+\sum_{T=1}^{\infty}\left( \int_{a+T-1}^{b+T-1} y'g(y')dy' - \int_{a+T}^{b+T} y'g(y')dy' \right)
	\nonumber \\&
	+\sum_{T=1}^{\infty}\left((1-T-b)G(b+T-1)+(b+T)G(b+T)\right)\nonumber\\
	&+\sum_{T=1}^\infty((a+T-1)G(a+T-1)-(a+T)G(a+T)) 
	\nonumber \\
	=&bG(b)-aG(a)-\int_a^b y'g(y')dy' +\int_a^b y'g(y')dy'
	-b(G(b)-1)+a(G(a)-1) \nonumber \\
	=&b-a\,,\nonumber
\end{align}
implying that $\pi((b-a))=b-a$. Hence, \eqref{pi.staionary.general} follows.
\hfill $\square$
\\

\noindent{\bf Proof of Theorem \ref{ergodic}.}
Consider the state space $\mathbb{S}=\mathbb{R}\times [0,1)$ of $M_t:=\{X_t,Y_t\}$ with Borel $\sigma$-algebra $\mathcal{S}$, the stationarity and geometric ergodicity are shown by checking the conditions of Theorem 5.1 in \cite{Stelzer2009} as follows:
\begin{enumerate}
	\item $\{M_t\}$ is a weak Feller chain, that is, $E(g(M_2)|M_1 = x)$ is continuous in $x \in \mathbb{S}$ for all bounded and continuous function $g: \mathbb{S} \to \mathbb{S}$.
	\item ${M_t}$ is $\mu$-irreducible, i.e., $\mu(A)>0$ implies $\sum_{n=1}^{\infty}P^n(x, A) >0$ for any Borel set $A\in\mathcal{S}$, where $\mu$ is some non-degenerate measure on ($\mathbb{S}$, $\mathcal{S}$) and 
	$P^n(\cdot, \cdot)$ is the $n$-step transition kernel of $M_t$.
	\item Denote the AR coefficient in TASS model at time $t$ as $\psi_t$.  There exists $\eta \in (0,1]$ and $c<1$ such that 
	$E\left(|\psi_t|^\eta | Y_1 = \delta \right) \leq c$ for any $\delta \in [0,1)$.
\end{enumerate}

The first condition is verified by directly calculating the expectation using \eqref{x2}.
The second condition is clear from the definition of the model.   
Moreover, under the assumption that $|\phi_j|<1$ for $j = 1,\ldots,m$, the third condition holds.
Hence, by Theorem 5.1 in \cite{Stelzer2009}, $\{X_t, Y_t\}$ is stationary and geometrically ergodic.
Finally, from Proposition 2 in \cite{Liebscher2005}, 
stationarity and geometric ergodicity is equivalent to $\beta$-mixing, which completes the proof.
\hfill $\square$

{\bf \noindent Proof of Lemma \ref{prob.property}.}   
Observed that if $Y_t$ is in the $i$th regime, then $X_t \vert Y_t \sim N\left(a_i,\frac{\sigma_i^2}{1-\phi_i^2}\right)\,$, which is the stationary distribution for the AR(1) model. Thus, 
\begin{eqnarray*}
	p(x_t \vert y_t)=\sqrt{\frac{1-\phi_i^2}{2\pi\sigma_i^2}}\mbox{exp}\left(-\frac{(1-\phi_i^2)(x_t-a_i)^2}{2\sigma_i^2}\right)\,.
\end{eqnarray*}
Moreover, if $Y_{t+1}$ is in the $j$th regime,  $X_{t+1}=\phi_j(X_t-a_j)+a_j+e_t$. 
Thus, we have 
\begin{eqnarray*}
	p(x_{t+1} \vert x_t, y_{t+1})=\frac{1}{\sqrt{2\pi\sigma_j^2}}\mbox{exp}\left(- \frac{(x_{t+1}-\phi_j(x_t-a_j)-a_j)^2}{2\sigma_j^2} \right)\,.
\end{eqnarray*}
Finally, we consider the conditional probability density function of $X_{t+2}$ given $X_t,Y_{t+1},Y_{t+2}$. If $Y_{t+1}$ is in the $j$th regime and $Y_{t+2}$ is in the $k$th regime, we have
\begin{align}
	X_{t+2}&=\phi_k(X_{t+1}-a_k)+a_k+\sigma_k e_{t+2} \,,\label{Xt2}\\ 
	X_{t+1}&=\phi_j(X_{t}-a_j)+a_j+\sigma_j e_{t+1} \,. \label{Xt1}
\end{align} 
Substitute \eqref{Xt1} into \eqref{Xt2}, we get 
\begin{eqnarray*}
	X_{t+2}=\phi_k\phi_j(X_t-a_j)+\phi_k(a_j-a_k)+a_k+\phi_k\sigma_j e_{t+1}+\sigma_k e_{t+2}\,.
\end{eqnarray*}
Therefore, $X_{t+2} \vert X_t,Y_{t+1},Y_{t+2} \sim N\big(\phi_k\phi_j(X_t-a_j)+\phi_k(a_j-a_k)+a_k,$ $\phi_k^2 \sigma_j^2+\sigma_k^2\big)\,$, and the conditional density function for $X_{t+2}$ given $X_t,Y_{t+1},Y_{t+2}$ is calculated as
\begin{align*}
	p(x_{t+2} \vert x_t, y_{t+1}, y_{t+2})=&\frac{1}{\sqrt{2\pi(\phi_k^2 \sigma_j^2+\sigma_k^2)}}\cdot\\
	&\mbox{exp}\left(-\frac{(x_{t+2}-\phi_k\phi_j(x_t-a_j)-\phi_k(a_j-a_k)-a_k)^2}{2(\phi_k^2 \sigma_j^2+\sigma_k^2)} \right)\,.
\end{align*}
\hfill $\square$ \\ 

\noindent {\bf \noindent Proof of Proposition \ref{prop:triple.p}.} For different regimes of $Y_t$, $X_t$ follows different distribution accordingly. 
In the two-regime TASS model, each of $Y_t$ belongs to either $[0,r_1)$ or $[r_1,1)$. 
Thus, the joint distribution of $X_t, X_{t+1}$ and $X_{t+2}$ can be partitioned into eight scenarios, given by
\begin{align}\label{joint.012}
	p(x_t, x_{t+1}, x_{t+2};\theta_m)=\sum_{i,j,k\in\{1,2\}} f_t(i,j,k)\,,
\end{align}
{\small where $f_t(i,j,k)=\int_{r_{i-1}}^{r_i} \int_{r_{j-1}}^{r_j} \int_{r_{k-1}}^{r_k} p(x_t \vert y_t) p(x_{t+1} \vert x_t, y_{t+1})p(x_{t+2} \vert x_{t+1}, y_{t+2})p(y_{t+2} \vert y_{t+1})$ $p(y_{t+1} \vert y_t)p(y_t)dy_{t+2}$ $dy_{t+1} dy_{t}$, 
	with $r_0=0$ and $r_2=1$}.

Given $i$, $j$ and $k$, the conditional density functions involving $x_t$'s do not involve $y_t$'s. Therefore, $p(x_t \vert y_t) p(x_{t+1} \vert x_t, y_{t+1})p(x_{t+2} \vert x_{t+1}, y_{t+2})$ can be taken out of the integrals. In other words, $f_t(i,j,k)$ can be factorized as $g_t(i,j,k)w_t(i,j,k)$, where $g_t(i,j,k)$ represents the density for the three consecutive observations with the parameters of $i$th, $j$th and $k$th regime, respectively, while $w_t(i,j,k)$ represents the probability that the three consecutive observations are in the $i$th, $j$th and $k$th regime. The explicit form of the $g_t(i,j,k)$ is stated in Proposition \ref{prop:triple.p}. It remains to derive $w_t(i,j,k)$'s by substituting  \eqref{conditionpr2} into  \eqref{joint.012}. Moreover, from Theorem \ref{stationary}, $p(y_t)=\pi(y_t)=1$. Hence, denote $\mathbf{G}(a)-\mathbf{G}(b)$ as $\tilde{\mathbf{G}}(a,b)$, $w_t(i,j,k)$'s can be evaluated as follows.
\allowdisplaybreaks

\begin{smaller}
	\begin{eqnarray}  
		&& w_t(1,1,1) \nonumber \\
		&=&\int_0^{r_1} \int_0^{r_1}  \int_0^{r_1} p(y_{t+2} \vert y_{t+1})p(y_{t+1} \vert y_t)p(y_t)dy_{t+2} dy_{t+1} dy_{t} \label{w.1} \\
		&=&\int_0^{r_1} \int_{0}^{y_{t+1}} \mathbf{G}(r_1-y_{t+1})
		\frac{\beta^\alpha}{\Gamma(\alpha)}(y_{t+1}-y_{t})^{\alpha-1} e^{-\beta(y_{t+1}-y_{t})} dy_t dy_{t+1} \nonumber  \\
		&&+ \sum_{T=1}^\infty
		\int_0^{r_1} \int_{0}^{y_{t+1}} \big(\tilde{\mathbf{G}}(r_1-y_{t+1}+T, T-y_{t+1}) 
		\frac{\beta^\alpha}{\Gamma(\alpha)}(y_{t+1}-y_{t})^{\alpha-1} \nonumber \\
		&&e^{-\beta(y_{t+1}-y_{t})}\big) dy_tdy_{t+1} \nonumber  \\
		&&+ \sum_{T=1}^\infty
		\int_0^{r_1} \int_{0}^{r_1} \mathbf{G}(r_1-y_{t+1})
		\frac{\beta^\alpha}{\Gamma(\alpha)}(y_{t+1}-y_{t}+T)^{\alpha-1}e^{-\beta(y_{t+1}-y_{t}+T)} dy_tdy_{t+1} \nonumber  \\
		&&+\sum_{T_1=1}^\infty\sum_{T_2=1}^\infty
		\int_0^{r_1} \int_{0}^{r_1} \big(\tilde{\mathbf{G}}(r_1-y_{t+1}+T_1, T_1-y_{t+1}) 
		\frac{\beta^\alpha}{\Gamma(\alpha)}(y_{t+1}-y_{t}+T_2)^{\alpha-1} \nonumber \\
		&&e^{-\beta(y_{t+1}-y_{t}+T_2)}\big) dy_tdy_{t+1} \nonumber  \\
		&=&\int_0^{r_1} \mathbf{G}(r_1-y_{t+1})\mathbf{G}(y_{t+1})dy_{t+1} \nonumber \\
		&&+\sum_{T=1}^\infty \int_0^{r_1} 
		\tilde{\mathbf{G}}(r_1-y_{t+1}+T, T-y_{t+1})
		\mathbf{G}(y_{t+1})dy_{t+1} \nonumber  \\
		&&+\sum_{T=1}^\infty \int_0^{r_1} \mathbf{G}(r_1-y_{t+1})
		\tilde{ \mathbf{G}}(y_{t+1}+T,y_{t+1}-r_1+T)dy_{t+1} \nonumber \\
		&&+\sum_{T_1=1}^\infty   \sum_{T_2=1}^\infty
		\int_0^{r_1} 
		\tilde{\mathbf{G}}(r_1-y_{t+1}+T_1, T_1-y_{t+1})
		\tilde{ \mathbf{G}}(y_{t+1}+T_2, y_{t+1}-r_1+T_2)dy_{t+1}\nonumber\,.
	\end{eqnarray}
\end{smaller}
Similarly,
\begin{small}
	\begin{eqnarray} \label{w.2}
		w_t(1,1,2) &=&\int_0^{r_1} \int_0^{r_1}  \int_{r_1}^1 p(y_{t+2} \vert y_{t+1})p(y_{t+1} \vert y_t)p(y_t)dy_{t+2} dy_{t+1} dy_{t}  \\
		&=&\sum_{T=0}^{\infty}
		\int_0^{r_1}\tilde{ \mathbf{G}}(1-y_{t+1}+T, r_1-y_{t+1}+T)\mathbf{G}(y_{t+1}) dy_{t+1} \nonumber \\
		&&+\sum_{T_1=0}^{\infty} \sum_{T_2=1}^{\infty}
		\int_0^{r_1} \big(\tilde{\mathbf{G}}(1-y_{t+1}+T_1,r_1-y_{t+1}+T_1)\cdot \nonumber \\
		&&\tilde{\mathbf{G}}(y_{t+1}+T_2,y_{t+1}+T_2-r_1)\big) dt_{t+1}\,,\nonumber  \\
		w_t(1,2,1) &=&\int_0^{r_1} \int_{r_1}^1  \int_0^{r_1} p(y_{t+2} \vert y_{t+1})p(y_{t+1} \vert y_t)p(y_t)dy_{t+2} dy_{t+1} dy_{t} \label{w.3}\\
		&=&\sum_{T_1=1}^{\infty} \sum_{T_2=0}^{\infty}
		\int_{r_1}^1
		\big(\tilde{\mathbf{G}}(r_1-y_{t+1}+T_1, T_1-y_{t+1}) \cdot \nonumber \\
		&&\tilde{\mathbf{G}}(y_{t+1}+T_2, y_{t+1}-r_1+T_2)\big) dy_{t+1}\,, \nonumber \\
		w_t(1,2,2) &=&\int_0^{r_1} \int_{r_1}^1  \int_{r_1}^1 p(y_{t+2} \vert y_{t+1})p(y_{t+1} \vert y_t)p(y_t)dy_{t+2} dy_{t+1} dy_{t} \label{w.4}\\
		&=&\sum_{T=0}^\infty
		\int_{r_1}^1 \mathbf{G}(1-y_{t+1})
		\tilde{\mathbf{G}}(y_{t+1}+T, y_{t+1}-r_1+T)dy_{t+1} \nonumber  \\
		&&+\sum_{T_1=1}^{\infty} \sum_{T_2=0}^{\infty}
		\int_{r_1}^1
		\big(\tilde{\mathbf{G}}(1-y_{t+1}+T_2, r_1-y_{t+1}+T_2)\cdot \nonumber \\ 
		&&\tilde{\mathbf{G}}(y_{t+1}+T_1, y_{t+1}-r_1+T_1)\big) dy_{t+1}\,,\nonumber 
	\end{eqnarray}
	\begin{eqnarray}
		w_t(2,2,2)&=&\int_{r_1}^1 \int_{r_1}^1  \int_{r_1}^1 p(y_{t+2} \vert y_{t+1})p(y_{t+1} \vert y_t)p(y_t)dy_{t+2} dy_{t+1} dy_{t} \hspace{-1cm} \label{w.5} \\
		&=&\int_{r_1}^1 \mathbf{G}(1-y_{t+1})
		\mathbf{G}(y_{t+1}-r_1)dy_{t+1} \nonumber \\
		&&+\sum_{T=1}^\infty \int_{r_1}^1 \mathbf{G}(1-y_{t+1})
		\tilde{\mathbf{G}}(y_{t+1}-r_1+T, y_{t+1}-1+T)dy_{t+1} \nonumber \\
		&&+ \sum_{T=1}^\infty \int_{r_1}^1
		\tilde{\mathbf{G}}(1-y_{t+1}+T, r_1-y_{t+1}+T)\mathbf{G}(y_{t+1}-r_1)dy_{t+1} \nonumber \\
		&&+ \sum_{T_1=1}^{\infty} \sum_{T_2=1}^{\infty} \int_{r_1}^1
		\big(\tilde{\mathbf{G}}(1-y_{t+1}+T_1, r_1-y_{t+1}+T_1)\cdot \nonumber \\
		&&\tilde{\mathbf{G}}(y_{t+1}-r_1+T_2, y_{t+1}-1+T_2) \big) dy_{t+1}\,,\nonumber  \\
		w_t(2,2,1)&=&\int_{r_1}^1 \int_{r_1}^1  \int_0^{r_1} p(y_{t+2} \vert y_{t+1})p(y_{t+1} \vert y_t)p(y_t)dy_{t+2} dy_{t+1} dy_{t} \hspace{-1cm} \label{w.6}\\
		&=&\sum_{T=1}^\infty \int_{r_1}^1
		\tilde{\mathbf{G}}(r_1-y_{t+1}+T, T-y_{t+1})\mathbf{G}(y_{t+1}-r_1)dy_{t+1}\nonumber \\
		&&+\sum_{T_1=1}^{\infty} \sum_{T_2=1}^{\infty} \int_{r_1}^1
		\big(\tilde{\mathbf{G}}(r_1-y_{t+1}+T_1, T_1-y_{t+1}) \cdot \nonumber \\&&\tilde{\mathbf{G}}(y_{t+1}-r_1+T_2, y_{t+1}-1+T_2)\big) dy_{t+1}\,, \nonumber 
	\end{eqnarray}
	\begin{eqnarray}\label{w.7}
		w_t(2,1,1) 
		&=&\int_{r_1}^1 \int_0^{r_1}  \int_0^{r_1} p(y_{t+2} \vert y_{t+1})p(y_{t+1} \vert y_t)p(y_t)dy_{t+2} dy_{t+1} dy_{t} \\
		&=&\sum_{T=1}^{\infty} \int_0^{r_1}
		\mathbf{G}(r_1-y_{t+1})
		\tilde{\mathbf{G}}(y_{t+1}-r_1+T,y_{t+1}-1+T)dy_{t+1} \nonumber \\
		& &+\sum_{T_1=1}^{\infty} \sum_{T_2=1}^{\infty} \int_0^{r_1}
		\big(\tilde{\mathbf{G}}(r_1-y_{t+1}+T_2, T_2-y_{t+1}) \cdot \nonumber \\
		&& \tilde{\mathbf{G}}(y_{t+1}-r_1+T_1, y_{t+1}-1+T_1)\big) dy_{t+1}\,, \nonumber  \\
		w_t(2,1,2)
		&=&\int_{r_1}^1 \int_0^{r_1}  \int_{r_1}^1 p(y_{t+2} \vert y_{t+1})p(y_{t+1} \vert y_t)p(y_t)dy_{t+2} dy_{t+1} dy_{t} 
		\hspace{-1cm} \label{w.8}\\
		&=&\sum_{T_1=0}^{\infty} \sum_{T_2=1}^{\infty} \int_0^{r_1}
		\big(\tilde{\mathbf{G}}(1-y_{t+1}+T_1, r_1-y_{t+1}+T_1) \cdot \nonumber \\&& \tilde{\mathbf{G}}(y_{t+1}-r_1+T_2, y_{t+1}-1+T_2)\big) dy_{t+1}\,.\nonumber 
	\end{eqnarray}
\end{small}
Finally, by substituting \eqref{x1}, \eqref{x2} \eqref{w.1} to \eqref{w.8} into \eqref{joint.012}, we arrive at the formula of the joint density $p(x_t, x_{t+1},x_{t+2};\theta_m)$.  \hfill $\square$



\begin{lemma} \label{lemma:derivatives}
	For any positive integer $m$,  $\log p(x_{t},x_{t+1},x_{t+2};\theta_m)$ is a twice continuously differentiable function of $\theta_m$. Furthermore, there exists an integrable function $g_m(x_t,x_{t+1},x_{t+2})$ such that  $$\sup_{{\theta_m} \in \Theta_m}\left|\frac{\partial^k}{\partial {\theta_m^k}}\log p(x_t,x_{t+1},x_{t+2};\theta_m)\right|<g_m(x_t,x_{t+1},x_{t+2})\,,$$ for $k=0,1,2$. 
\end{lemma} 

{\bf \noindent Proof of Theorem \ref{thm.consistency}.} We divide the proof into two parts. First, we show the consistency and asymptotic normality of the parameter estimator $\hat{\theta}_{m_o}$ (denote as $\hat{\theta}$ below for simplicity) when the true number of regimes $m=m_o$ is known. Second, we show that $m_0$ can be consistently estimated by BIC.

\noindent \textbf{1. Consistency and asymptotic normality of $\hat{\theta}_{m_o}$}:
Denote the normalized composite log-likelihood function as 
$L_n(\theta_m)= \frac{1}{n-2} \sum_{t=1}^{n-2} \log p(x_t, x_{t+1}, x_{t+2};\theta_m)$ where
$\theta_m$ is the parameter of a TASS$(m)$ model. Let $L(\theta_m)=E_{{\theta}^o}(\log p(x_1, x_{2}, x_{3}$ $;\theta_m))$ 
with the expectation evaluated at the true parameter value.  
The estimator $\hat{\theta}_m$ defined in Section  
\ref{sec:model.select} can also be expressed as $\hat{{\theta}}_m=\argmax_{{\theta}_m} L_n({\theta}_m)$. 
Note that by stationarity, $E_{{\theta}^o}(L_n({\theta_m}))=L({\theta_m})$ for any $n\geq 3$ and ${\theta_m}$. 
From Theorem \ref{ergodic} and Lemma \ref{lemma:derivatives}, by the standard uniform LLN, we have the uniform convergence result such that
\begin{eqnarray}\label{eq.u.conv}
	\sup_{{\theta}\in \Theta_{m_o}}|L_n({\theta})-L({\theta})| \to_{a.s.} 0\,,
\end{eqnarray} 
as $n \to \infty$. Now, by definition of the maximum likelihood estimator, $L_n({\theta}^o) \leqslant L_n(\hat{{\theta}})$. Also, by Jensen's Inequality, $L(\hat{{\theta}})\leqslant L({\theta}^o)$. Thus, we have 
\begin{align*}
	0\geq L(\hat{\theta})-L(\theta^o)&=L(\hat\theta)-L_n(\hat\theta)+L_n(\hat\theta)-L_n(\theta^o)+L_n(\theta^o)-L(\theta^o)\\
	&\geq L(\hat\theta)-L_n(\hat\theta)+L_n(\theta^o)-L(\theta^o),
\end{align*}
implying that 
$|L(\hat{{\theta}})-L({\theta}^o) |\leqslant  |L_n({\theta}^o)-L({\theta}^o)|+ |L_n(\hat{{\theta}}) - L(\hat{{\theta}})| $. 
By \eqref{eq.u.conv}, $|L(\hat{{\theta}})-L({\theta}^o)|\overset{a.s.}{\to} 0$ as $n \to \infty$. Thus,  $\hat{{\theta}} \overset{a.s.}{\to} {\theta}^o$ since $L({\theta})$ is smooth and has a unique maximizer at $\theta^o$ and $\Theta_{m_o}$ is compact. The asymptotic normality of $\hat{\theta}$ follows from the standard arguments based on a Taylor expansion of $L_n(\hat{\theta})$ around $\theta^o$ and the central limit theorem.

\noindent\textbf{2. Consistency of $\hat{m}$}: We first show that $m_o$ cannot be underestimated by BIC under the assumption stated in Theorem \ref{thm.consistency}. Given Theorem \ref{ergodic} and Lemma \ref{lemma:derivatives}, by the standard uniform LLN, we have that for any $m\leq m_o$,
\begin{align*}
	\sup_{\theta_m \in \Theta_m} \left|L_n(\theta_m)-L(\theta_m) \right| \to_{a.s.} 0.
\end{align*}
Together with the assumption, this implies that for any given $m\leq m_o$, we have that $\hat{\theta}_m\to \theta_m^*$, in other words, $\hat{\theta}_m$ is a consistent estimator for $\theta_m^*.$ Thus, by the uniform LLN, we have for $m<m_o$,
\begin{align*}
	\frac{1}{n}\text{BIC}(m)-\frac{1}{n}\text{BIC}(m_o) \to_{a.s.} -(L(\theta_m^*)-L(\theta^o) ) >0,
\end{align*}
where the inequality follows from information inequality. Thus we have $P(\hat{m}<m_o)\to 0$ as $n\to \infty.$

We now turn to the proof of $P(\hat{m}>m_o)\to 0$, which is more involved due to the over-parametrization issue, where we need to show that the increase of log-likelihood $L_n(\theta_m)$ brought by over-parametrization is less than $O(\log n/n)$ in probability.

Denote $\Psi_i=(a_i,\phi_i,\sigma_i)$. Given the number of states $m$, the parameter to be estimated can be expressed as $\theta_m=\{\mathbf{r}_m, (\alpha, \beta), \{\Psi_i\}_{i=1}^m\}$, where $\mathbf{r}_m=(r_1,\cdots,r_{m-1})$. Also, the true parameter is $\theta^o=\{\mathbf{r}_m^o,(\alpha^o,\beta^o), \{\Psi_i^o\}_{i=1}^{m_o}\}$, where $(r_1^o,\cdots,r_{m_o-1}^o)$ with  $\Psi_i^o=\{a_i^o,\phi_i^o,\sigma_i^o\}$. 
Note that for $m_o<m\leq M$, due to over-parametrization, there exists $\theta_m^*\in\Theta_m$ such that $$\log p(x_t,x_{t+1},x_{t+2};\theta_m^*)\equiv\log p(x_t,x_{t+1},x_{t+2};\theta^o).$$ 
To see this, 
WLOG, assume that for $\theta_m^*$, we have $(r_1,r_2,\cdots,r_{m_o-1})=(r_1^o,\cdots,$ $r_{m_o-1}^o)$. Thus, for $\theta_m^*$ such that $(\alpha,\beta)=(\alpha^o,\beta^o)$, $\Psi_i=\Psi_i^o, i=1,\cdots,m_o$ and $\Psi_i=\Psi_{m_o}, i=m_o+1,\cdots, m$, it holds that $\log p(x_t,x_{t+1},x_{t+2};\theta_m^*)=\log p(x_t,x_{t+1},x_{t+2};\theta^o)$ and thus $L_n(\theta_m^*)=L_n(\theta^o)$, regardless of the values of $(r_{m_o},\cdots,r_{m-1})$. The same conclusion holds for any $\theta^*_m$ such that $(r_1^o,\cdots,r_{m_o-1}^o) $ $\subset \mathbf{r}_m$ and $(\alpha,\beta)=(\alpha^o,\beta^o)$, $\Psi_j=\Psi_{i(j)}^o, j=1,\cdots,m$, where $i(j)=i$ such that $(r_{j-1},r_j)\subset (r_{i-1}^o-\epsilon/4, r_i^o+\epsilon/4)$.

Using the same argument for consistency as before, we can readily show that for any $m>m_o$, it holds that
\begin{align*}
	&\max_{i=1,\cdots,m_o-1}\min_{j=1,\cdots,m-1} |r_i^o-\hat{r}_j|\to_{a.s.} 0, \\
	&(\hat{\alpha},\hat{\beta})\to_{a.s.} (\alpha^o,\beta^o) \text{ and }
	\hat{\Psi}_j -\Psi_{i(j)}^o \to_{a.s.} 0, ~j=1,\cdots,m,
\end{align*}
where $i(j)=i$ such that $(\hat{r}_{j-1},\hat{r}_j)\subset (r_{i-1}^o-\epsilon/4, r_i^o+\epsilon/4)$.

By the consistency result of $\hat{\theta}_m$, WLOG, we assume that for $\hat{\theta}_m$, we have $(\hat{r}_1,\cdots, \hat{r}_{m_o-1})\to_{a.s.} \mathbf{r}^o$. Denote $\underline{\hat{r}}_m=(\hat{r}_{m_o},\cdots,\hat{r}_{m-1})$ and $\underline{\hat{\theta}}_m=\hat{\theta}_m\setminus \underline{\hat{r}}_m$. On each $\omega$ of the probability space $\Omega$, for any subsequence of $\hat\theta_m$, due to the boundedness of $[0,1]$, we can find a further subsequence such that $\underline{\hat{r}}_m \to \underline{r}_m^*= (r_{m_o}^*,\cdots, r_{m-1}^*)$. On that subsequence of $\hat\theta_m$, by a standard Taylor expansion of $\underline{\hat{\theta}}_m$ around $\underline{\theta}_m^*$ and $\underline{\hat{r}}_m$ around $\underline{{r}}_m^*$, we have that    
\begin{align}
	0=& \begin{bmatrix}
		\frac{\partial}{\partial \underline{\theta}_m}L_n(\hat{\theta}_m)\\
		\frac{\partial}{\partial \underline{r}_m}L_n(\hat{\theta}_m)
	\end{bmatrix}=
	\begin{bmatrix}
		\frac{\partial}{\partial \underline{\theta}_m}L_n(\underline{\hat{\theta}}_m,\underline{\hat{r}}_m)\\
		\frac{\partial}{\partial \underline{r}_m}L_n(\underline{\hat{\theta}}_m,\underline{\hat{r}}_m)
	\end{bmatrix}\nonumber\\
	=&
	\begin{bmatrix}
		\frac{\partial}{\partial \underline{\theta}_m}L_n(\underline{\theta}_m^{*},\underline{r}_m^*)\\
		0\end{bmatrix}
	+ 
	\begin{bmatrix}
		\frac{\partial^2}{\partial \underline{\theta}_m^{2}}L_n(\underline{\tilde\theta}_m,\underline{\tilde r}_m) &
		\frac{\partial^2}{\partial \underline{\theta}_m \partial \underline{r}_m}L_n(\underline{\tilde\theta}_m,\underline{\tilde r}_m)\\
		\frac{\partial^2}{\partial \underline{r}_m \partial \underline{\theta}_m}L_n(\underline{\tilde\theta}_m,\underline{\tilde r}_m)&
		\frac{\partial^2}{\partial \underline{r}_m^{2}}L_n(\underline{\tilde\theta}_m,\underline{\tilde r}_m)
	\end{bmatrix}
	\begin{bmatrix}
		\underline{\hat{\theta}}_{m}-\underline{\theta}_m^{*}\\
		\underline{\hat{r}}_{m}-\underline{r}_m^{*}
	\end{bmatrix},\label{eq:Taylor}
\end{align}
where $(\underline{\tilde\theta}_m,\underline{\tilde r}_m)$ is between $(\underline{\theta}_m^{*},\underline{r}_m^*)$ and $(\underline{\hat{\theta}}_m,\underline{\hat{r}}_m)$. Note that it is easy to see that $\frac{\partial}{\partial \underline{r}_m}L_n(\underline{\theta}_m^{*},\underline{r}_m^*)=0$, $\frac{\partial}{\partial \underline{\theta}_m}L(\underline{\theta}_m^{*},\underline{r}_m^*)=0$, $\frac{\partial}{\partial \underline{r}_m}L(\underline{\theta}_m^{*},\underline{r}_m^*)=0$, $\frac{\partial^2}{\partial \underline{r}_m \partial \underline{\theta}_m}L(\underline{\theta}^*_m,\underline{r}^*_m)=0$ and $\frac{\partial^2}{\partial^2 \underline{r}_m}L(\underline{\theta}^*_m,\underline{r}^*_m)=0$.

{
	By virtue of consecutive tuple likelihood, the first order derivative of the log three-tuple density is a measurable transformation of three consecutive pairs of $\{X_t,Y_t\}$. Thus, 
	$\frac{\partial}{\partial \underline{r}_m}L_n$ is $\beta$-mixing with geometric rate by Theorem \ref{ergodic}, and 
	satisfies the law of iterated logarithm by \citet{Rio:1995}. 
	Together with uniform LLN of the second order derivatives of $L_n$}, the Taylor expansion in \eqref{eq:Taylor} implies that on the further subsequence of $\hat{\theta}_m$, $|\underline{\hat{\theta}}_{m}-\underline{\theta}_m^{*}|=O(\sqrt{\frac{\log\log n}{n}}) \underset{\sim}{<}\frac{\log\log n}{\sqrt n}$.
Thus, by Lemma \ref{lemma:subsub} below, we have $\limsup_n |\underline{\hat{\theta}}_{m}-\underline{\theta}_m^{*}|/( \log\log n/\sqrt{n}) <1$ a.s.

With another Taylor expansion of $\underline{\hat{\theta}}_m$ around $\underline{\theta}_m^*$, we have that
\begin{align*}
	L_n(\hat{\theta}_m)=&L_n(\underline{\theta}_m^{*},\hat{r}_{m_o+1},\cdots,\hat{r}_{m}) \nonumber \\
	&- \frac{1}{2}(\underline{\hat{\theta}}_{m}-\underline{\theta}_m^{*})^\top\frac{\partial^2}{\partial \underline{\theta}_m^{2}}L_n(\underline{\tilde\theta}_m,\hat{r}_{m_o+1},\cdots,\hat{r}_{m}) (\underline{\hat{\theta}}_{m}-\underline{\theta}_m^{*})\\
	=&L_n(\theta^o)+O_p((\log\log n)^2/n) = L_n(\hat{\theta}_{m_o})+O_p((\log\log n)^2/n).
\end{align*}
Thus, we have $P(\min_{m_o+1\leq m\leq M}\text{BIC}(m)>\text{BIC}(m_o))\to 1$ and thus $P(\hat{m}=m_o)\to 1.$ \hfill $\square$ 


\begin{lemma}\label{lemma:subsub}
	Denote $C$ as an arbitrary constant. For any sequence $\{X_n\}$, if every subsequence of $\{X_n\}$ has a further subsequence such that its limsup is $\leq C$, then we have that $\limsup X_n \leq 2C.$
\end{lemma}
\begin{proof}
	Assume that $\limsup X_n > 2C$, then for any $k>1$, we can find a $n_k>k$ such that $X_{n_k}>1.5C$, then there is no subsequence of $X_{n_k}$ with limsup $\leq C$. Thus by contradiction, we have $\limsup X_n \leq 2C.$
\end{proof}

\

{\bf \noindent Proof of Theorem \ref{coro.TASS.MAP}.}
To prove the theorem, it suffices to verify the four conditions in Theorem 1 of \citet{bootstrapfilter2013}:
\begin{enumerate}
	\item[1)] the observed sequence $X_{1:n}=x_{1:n}$ is fixed.
	\item[2)] the likelihood $l_t(y_{1:t})=p(x_t \vert y_{1:t}, x_{1:t-1})$ is a bounded function of $y_{1:t}$ for each $t=1, \ldots,n$.
	\item[3)] the integral of the likelihood $l_t(y_{1:t})$ with respect to the measure $p(y_{1:t}$ $\vert x_{1:t-1})dy_{1:t}$ is positive, i.e., $\int l_t(y_{1:t})p(y_{1:t} \vert x_{1:t-1})dy_{1:t}>0$.
	\item[4)] the maximizer of the conditional distribution $p(y_{1:t}\vert x_{1:t})$ exists and $p(y_{1:n} \vert x_{1:n})$ is continuous at its global maxima.
\end{enumerate}
The first condition is automatically fulfilled. 
By the calculation of the density functions in Section \ref{properties}, 
\begin{eqnarray*}
	l_t(y_{1:t})&=&p(x_t \vert y_{1:t}, x_{1:t-1})=p(x_t \vert y_{t}, x_{t-1})\\
	&=&\frac{1}{\sqrt{2\pi\sigma_j^2}}\mbox{exp}\left(- \frac{(x_{t}-\phi_j(x_{t-1}-a_j)-a_j)^2}{2\sigma_j^2} \right)\,,
\end{eqnarray*}
given $y_t$ is in the $j$th regime, which is a normal density. Thus condition 2 is proved. 

Similarly, from above calculation, $l_t(y_{1:t})>0$. On the other hand, 
\begin{eqnarray*}
	p(y_{1:t} \vert x_{1:t}) \propto  p(y_{1:t}, x_{1:t}) = \prod_{i=2}^t p(x_i \vert x_{i-1}, y_i) p(x_1\vert y_1) \prod_{i=2}^t p(y_i \vert y_{i-1})p(y_1)
\end{eqnarray*}
is positive since each component in the product is positive, where the  explicit form can be found in section \ref{properties}. Therefore $\int l_t(y_{1:t})p(y_{1:t} \vert x_{1:t-1})dy_{1:t} $ is positive. Condition 3 is proved.

Lastly, $p(y_{1:n} \vert x_{1:n})=\prod_{i=2}^n p(x_i \vert x_{i-1}, y_i) p(x_1\vert y_1) \prod_{i=2}^n p(y_i \vert y_{i-1})p(y_1)$ is continuous and differentiable. Thus, all conditions have been verified.
\hfill $\square$ 

\




{\bf \noindent Proof of Lemma \ref{lemma:derivatives}.} The lemma is verified by explicitly analyzing the first and second order partial derivatives of $\log p_{\mathbf{\theta}}(x_t, x_{t+1}, x_{t+2};\theta_m)$, which will be denoted as $\log p(x_t, x_{t+1}, x_{t+2})$ below, with respect to each parameter.

Note that
\begin{eqnarray*}
	\frac{\partial  \log p(x_t, x_{t+1}, x_{t+2})}{\partial \theta_m}
	= \frac{1}{p(x_t, x_{t+1}, x_{t+2})}\frac{\partial p(x_t, x_{t+1}, x_{t+2})}{\partial \theta_m}\,,
\end{eqnarray*}
and
\begin{eqnarray*}
	&&\frac{\partial^2  \log p(x_t, x_{t+1}, x_{t+2})}{\partial \theta_m^2} \nonumber \\
	&=& \frac{1}{p^2(x_t, x_{t+1}, x_{t+2})}\cdot \nonumber \\
	&&\left(\frac{\partial^2 p(x_t, x_{t+1}, x_{t+2})}{\partial \theta_m^2}p(x_t, x_{t+1}, x_{t+2}) - \left(\frac{\partial p(x_t, x_{t+1}, x_{t+2})}{\partial \theta_m} \right)^2  \right) \,.
\end{eqnarray*}
Therefore, it suffices to check the boundedness of $p_{\mathbf{\theta}}(x_t, x_{t+1}, x_{t+2})$ as well as its first and second derivatives.

The boundedness of $p(x_t, x_{t+1}, x_{t+2})=\sum_{i,j,k=1,2} g_t(i,j,k) w_t(i,j,k)$ can be shown by checking the boundedness of $g_t(i,j,k)$ and $w_t(i,j,k)$, respectively. $g_t(i,j,k)$ is the product of normal densities, while $w_t(i,j,k)$ represents a set of probabilities bounded by 1. Thus, boundedness of $p(x_t, x_{t+1}, x_{t+2})$ is checked.

Next, we check the boundedness of the first partial derivative of $p(x_t, x_{t+1}, x_{t+2})$ w.r.t.\ each parameter in the vector $\mathbf{\theta}_m=( \phi_1, a_1, \sigma_1,\ldots,\phi_m, a_m, \sigma_m, r_1,\ldots,r_{m-1},$ $ \alpha, \beta ),$ for $i=1,\ldots,m, j=1,\dots,m-1$. Without loss of generality, we check the boundedness of derivatives for the 2-regime TASS model described in Section \ref{sec:cpl}. For the cases of more regimes and higher AR order, the calculations are similar. 

Observe that $g_t(i,j,k)$ only depends on parameters $\phi_1, \phi_2, \sigma_1, \sigma_2, a_1, a_2$, and $w_t({i,j,k})$ only depends on $\alpha, \beta$ and $r_1$, respectively. 
Moreover, $g_t(i,j,k)$ is uniformly bounded because the density function of normal distribution is bounded. On the other hand, $w_t(i,j,k)$ is bounded since it represents the probability for the specific events, which are within 0 and 1. Therefore, it suffices to consider the partial derivatives of $g_t(i,j,k)$ and $w_t(i,j,k)$ separately.  

For the partial derivatives of $w_t(i,j,k)$ with respect to $r_1$, we explicitly calculate the partial derivatives.
First, for $w_t(1,1,1)$, there are four components. We consider a typical component as follows. 
\begin{align*}
	&\frac{\partial}{\partial r_1}\left[ \sum_{T=1}^\infty \int_0^{r_1} \mathbf{G}_{\alpha, \beta}(r_1-y_{t+1})
	\tilde{\mathbf{G}}_{\alpha, \beta}(y_{t+1}+T,y_{t+1}-r_1+T)dy_{t+1} \right] \\
	=&\sum_{T=1}^{\infty} \Big(\int_0^{r_1} \{\mathbf{g}_{\alpha, \beta}(r_1-y_{t+1}) 
	\tilde{\mathbf{G}}_{\alpha, \beta}(y_{t+1}+T,y_{t+1}-r_1+T) \nonumber \\
	&+\mathbf{G}_{\alpha, \beta}(r_1-y_{t+1})\mathbf{g}_{\alpha, \beta}(y_{t+1}-r_1+T)\}dy_{t+1}) \Big)\\
	= &\sum_{T=1}^{\infty} [\tilde{\mathbf{G}}_{2\alpha, \beta}(r_1+T, T-r_1)+\tilde{\mathbf{G}}_{2\alpha, \beta}(T,T-2r_1)]\,,  
\end{align*}
and
\begin{align*}
	\begin{split}
		&\frac{\partial^2}{\partial r_1^2}\left[ \sum_{T=1}^\infty \int_0^{r_1} \mathbf{G}_{\alpha, \beta}(r_1-y_{t+1})
		\tilde{\mathbf{G}}_{\alpha, \beta}(y_{t+1}+T,y_{t+1}-r_1+T)dy_{t+1} \right] \\
		=&\sum_{T=1}^\infty [\mathbf{g}_{2\alpha, \beta}(r_1+T)+\mathbf{g}_{2\alpha, \beta}(T-r_1)
		+2\mathbf{g}_{2\alpha, \beta}(T-2r_1)],
	\end{split}
\end{align*}
where $\mathbf{g}_{\alpha, \beta}(\cdot)$ 
and $\mathbf{G}_{\alpha, \beta}(\cdot)$ denote the density and cumulative distribution function of Gamma($\alpha, \beta$), respectively, and $\tilde{\mathbf{G}}_{\alpha, \beta}(a, b)$ denotes $\mathbf{G}_{\alpha, \beta}(a)-\mathbf{G}_{\alpha, \beta}(b)$. It can be checked that the above components are bounded. Similarly, all first and second order partial derivatives of $w_t(i,j,k)$'s with respect to $r_1$ can be shown to be bounded.

Next, we derive partial derivatives of $w_t(1,1,1)$ with respect to $\alpha$ and $\beta$. Boundedness of derivatives for other $w_t(i,j,k)$'s can be proved similarly. 
Note that for any $x>0$, denote $\gamma$ as the Euler-Mascheroni constant, we have
\begin{eqnarray} \label{gamma.deri}
	\begin{split}
		\frac{\partial}{\partial \alpha}\mathbf{G}_{\alpha, \beta}(x)
		=&\frac{\partial}{\partial \alpha} \left[\frac{1}{\Gamma(\alpha)} \int_0^{\beta x}t^{\alpha -1}e^{-t}dt \right] \\
		=&\frac{\alpha -1}{\Gamma (\alpha)}\int_0^{\beta x} t^{\alpha-2}e^{-t} dt \nonumber \\
		&- \frac{1}{\Gamma(\alpha)}\left(-\gamma+\sum_{n \geqslant 1}\left(\frac{1}{n}-\frac{1}{n+\alpha -1}\right)\right)\int_0^{\beta x}t^{\alpha -1}e^{-t}dt \\
		=&\mathbf{G}_{\alpha -1, \beta}(x)+\left(\gamma -\sum_{n \geqslant 1} \left(\frac{1}{n}-\frac{1}{n+ \alpha -1} \right) \right)\mathbf{G}_{\alpha, \beta}(x)\,.
	\end{split}
\end{eqnarray}
It is obvious that the above partial derivative is bounded.
Now, using \eqref{gamma.deri}, we analyze the partial derivatives of one typical component of $w_t(1,1,1)$ with respect to $\alpha$.

	

\begin{align*}
	\begin{split}
		&\frac{\partial}{\partial \alpha} \Bigg[
		\sum_{T_1=1}^\infty   \sum_{T_2=1}^\infty
		\int_0^{r_1} \Big[
		\tilde{\mathbf{G}}_{\alpha, \beta}(r_1-y_{t+1}+T_1, T_1-y_{t+1})\cdot \nonumber \\
		& \tilde{\mathbf{G}}_{\alpha, \beta}(y_{t+1}+T_2, y_{t+1}-r_1+T_2) \Big] dy_{t+1} \Bigg]  \\
		=&\int_0^{r_1} \sum_{T_1=1}^\infty \sum_{T_2=1}^\infty \Big\{\tilde{\mathbf{G}}_{\alpha -1, \beta}(r_1-y_{t+1}+T_1, T_1-y_{t+1})
		\tilde{\mathbf{G}}_{\alpha , \beta}(y_{t+1}+T_2, y_{t+1}-r_1+T_2)\\
		&+ \tilde{\mathbf{G}}_{\alpha , \beta}(r_1-y_{t+1}+T_1, T_1-y_{t+1})
		\tilde{\mathbf{G}}_{\alpha -1 , \beta}(y_{t+1}+T_2, y_{t+1}-r_1+T_2)  \\
		&+\left(\gamma -\sum_{n \geqslant 1} \left(\frac{1}{n}-\frac{1}{n+ \alpha -1} \right) \right)
		\tilde{\mathbf{G}}_{\alpha, \beta}(r_1-y_{t+1}+T_1, T_1-y_{t+1}) \cdot \\
		&\tilde{\mathbf{G}}_{\alpha, \beta}(y_{t+1}+T_2, y_{t+1}-r_1+T_2) 
		\Big\}d y_{t+1} \,.
	\end{split}
\end{align*}
It can be seen that the integrand in the above integral is bounded by a constant $C$. Hence the partial derivatives with respect to $\alpha$ is bounded by $r_1C$.  
The second order derivative of $w_t(1,1,1)$ shares the same structure with above integral.
Therefore, the boundedness can be similarly proved.

For the partial derivatives with respect to $\beta$, note that
\begin{eqnarray}
	\frac{\partial}{\partial \beta}\mathbf{G}_{\alpha, \beta}(x)
	=&\frac{\partial}{\partial \beta} \left[\frac{1}{\Gamma(\alpha)} \int_0^{\beta x}t^{\alpha -1}e^{-t}dt \right] 
	\nonumber \\
	=& \frac{1}{\Gamma(\alpha)}x (\beta x)^{\alpha -1}e^{-\beta x} \nonumber \\
	= &\frac{\alpha}{\beta ^2}\mathbf{g}_{\alpha +1, \beta}(x)\,.
\end{eqnarray}
Similar to the derivation of partial derivative with respect to $\alpha$, we illustrate the partial derivatives of a typical component of $w_t(1,1,1)$ with respect to $\beta$ as follows. The boundedness of partial derivatives of other $w_t(i,j,k)$'s can be shown similarly.
\begin{align*}
	\begin{split}
		&\frac{\partial}{\partial \beta}\left[ \sum_{T=1}^\infty \int_0^{r_1} 
		\tilde{\mathbf{G}}_{\alpha, \beta}(r_1-y_{t+1}+T,T-y_{t+1})\mathbf{G}_{\alpha, \beta}(y_{t+1})dy_{t+1} \right] \\
		=&\frac{\alpha}{\beta ^2}\sum_{T=1}^\infty \int_0^{r_1}
		\Big( \mathbf{G}_{\alpha, \beta}(y_{t+1})\tilde{\mathbf{g}}_{\alpha + 1, \beta}(r_1 -y_{t+1}+T,T-y_{t+1})\cdot  \\
		&+ \mathbf{g}_{\alpha +1, \beta}(y_{t+1})
		\tilde{\mathbf{G}}_{\alpha, \beta}(r_1 -y_{t+1}+T, T-y_{t+1})  \Big) dy_{t+1} \\
		=&\frac{2\alpha}{\beta ^2} \sum_{T=1}^\infty
		\left\{
		\tilde{\mathbf{G}}_{2\alpha +1, \beta}(T+r_1, T) 
		-\tilde{\mathbf{G}}_{2\alpha +1, \beta}(T,T-r_1) 
		\right\}\,,\\
		&\frac{\partial^2}{\partial \beta^2}\left[ \sum_{T=1}^\infty \int_0^{r_1} 
		\tilde{\mathbf{G}}_{\alpha, \beta}(r_1-y_{t+1}+T,T-y_{t+1})\mathbf{G}_{\alpha, \beta}(y_{t+1})dy_{t+1} \right] \\
		=& \frac{2\alpha(2\alpha+1)}{\beta^4} \sum_{T=1}^\infty
		(\tilde{\mathbf{g}}_{2\alpha+2, \beta}(T+r_1, T)-
		\tilde{\mathbf{g}}_{2\alpha+2, \beta}(T, T-r_1))\,.
	\end{split}
\end{align*}

On the other hand, the first and second order partial derivatives of $g_t(i,j,k)$'s with respect to $\phi_1, \phi_2, \sigma_1, \sigma_2, a_1$ and $a_2$ can be calculated trivially. The boundedness can readily be shown. This completes the proof. \hfill $\square$

\newpage
\bibliographystyle{apalike}
\bibliography{paper-ref}

\begin{thebibliography}{}

\bibitem[Adams and MacKay, 2007]{Ryan2007}
Adams, R.~P. and MacKay, D.~J. (2007).
\newblock Bayesian online changepoint detection.
\newblock Technical report, Cambridge, UK.

\bibitem[Aroian and Levene, 1950]{Leo1950}
Aroian, L.~A. and Levene, H. (1950).
\newblock The effectiveness of quality control charts.
\newblock {\em Journal of the American Statistical Association},
  45(252):520--529.

\bibitem[Aue and Horv\'{a}th, 2013]{Aue-Horvath13}
Aue, A. and Horv\'{a}th (2013).
\newblock Structural breaks in time series.
\newblock {\em Journal of Time Series Analysis}, 34:1--16.

\bibitem[Beaulieu et~al., 2012]{Climate2012}
Beaulieu, C., Chen, J., and Sarmiento, J. (2012).
\newblock Change-point analysis as a tool to detect abrupt climate variations.
\newblock {\em Philosophical transactions. Series A, Mathematical, physical,
  and engineering sciences}, 370:1228--1249.

\bibitem[Braun et~al., 2000]{Muller2000}
Braun, J.~V., Braun, R.~K., and Muller, H.~G. (2000).
\newblock Multiple changepoint fitting via quasilikelihood, with application to
  dna sequence segmentation.
\newblock {\em Biometrika}, 87(2):301--314.

\bibitem[Chan et~al., 2015]{ChanYauZhang2015}
Chan, N.~H., Yau, C.~Y., and Zhang, R.~M. (2015).
\newblock Lasso estimation of threshold autoregressive models.
\newblock {\em Journal of Econometrics}, 189(2):285--296.

\bibitem[Cho and Fryzlewicz, 2015]{Cho2015}
Cho, H. and Fryzlewicz, P. (2015).
\newblock Multiple-change-point detection for high dimensional time series via
  sparsified binary segmentation.
\newblock {\em Journal of the Royal Statistical Society. Series B},
  77:475--507.

\bibitem[Choi et~al., 2008]{Choi2012}
Choi, H., Ombao, H., and Ray, B. (2008).
\newblock Sequential change-point detection methods for nonstationary time
  series.
\newblock {\em Technometrics}, 50(1):40--52.

\bibitem[Costa et~al., 2016]{Environment2016}
Costa, M., Goncalves, A.~M., and Teixeira, L. (2016).
\newblock Change-point detection in environmental time series based on the
  informational approach.
\newblock {\em Electronic Journal of Applied Statistical Analysis},
  9(2):267--296.

\bibitem[Cs\"{o}rg\H{o} and Horv\'{a}th, 1997]{Csorgo-Horvath97}
Cs\"{o}rg\H{o}, M. and Horv\'{a}th, L. (1997).
\newblock {\em Limit Theorems in Change-Point Analysis}.
\newblock Wiley, New York.

\bibitem[Davis and Yau, 2011]{Davis-Yau11}
Davis, R.~A. and Yau, C.~Y. (2011).
\newblock Comments on pairwise likelihood in time series models.
\newblock {\em Statistica Sinica}, 21(1):255--278.

\bibitem[Enikeeva and Harchaoui, 2019]{Enikeeva2019}
Enikeeva, F. and Harchaoui, Z. (2019).
\newblock High-dimensional change-point detection under sparse alternatives.
\newblock {\em Annals of Statistics}, 47:2051--2079.

\bibitem[Fearnhead and Liu, 2007]{Fearnhead2007}
Fearnhead, P. and Liu, Z. (2007).
\newblock Online inference for multiple changepoint problems.
\newblock {\em Journal of the Royal Statistical Society: Series B},
  69(4):589--605.

\bibitem[Godsill et~al., 2001]{MAP2001}
Godsill, S., Doucet, A., and West, M. (2001).
\newblock Maximum a posteriori sequence estimation using monte carlo particle
  filters.
\newblock {\em Annals of the Institute of Statistical Mathematics},
  53(1):82--96.

\bibitem[Gordon et~al., 1993]{Gordon1993}
Gordon, N., Salmond, D., and Smith, A. (1993).
\newblock Novel approach to nonlinear/non-gaussian bayesian state estimation.
\newblock {\em IEE Proceedings F - Radar and Signal Processing},
  140(2):107--113.

\bibitem[Jandhyala et~al., 2013]{Jandhyala-et-al13}
Jandhyala, V., Fotopoulos, S., MacNeill, I., and Liu, P. (2013).
\newblock Inference for single and multiple change points in time series.
\newblock {\em Journal of time series analysis}, 34(4):423--446.

\bibitem[Kaplan and Shishkin, 2000]{Kaplan2000}
Kaplan, A. and Shishkin, S. (2000).
\newblock Application of the change-point analysis to the investigation of the
  brain's electrical activity.
\newblock In Brodsky, B.~E. and Darkhovsky, B.~S., editors, {\em Non-Parametric
  Statistical Diagnosis: Problems and Methods}, chapter~7, pages 333--388.
  Springer, Dordrecht.

\bibitem[Liebscher, 2005]{Liebscher2005}
Liebscher, E. (2005).
\newblock Towards a unified approach for proving geometric ergodicity and
  mixing properties of nonlinear autoregressive processes.
\newblock {\em Journal of Time Series Analysis}, 26:669--689.

\bibitem[Ma and Yau, 2016]{MaYau2016}
Ma, T.~F. and Yau, C.~Y. (2016).
\newblock A pairwise likelihood-based approach for change-point detection in
  multivariate time series models.
\newblock {\em Biometrika}, 103(2):409--421.

\bibitem[Matteson and James, 2014]{Matteson2014}
Matteson, D.~S. and James, N.~A. (2014).
\newblock A nonparametric approach for multiple change point analysis of
  multivariate data.
\newblock {\em Journal of the American Statistical Association}, 109:334--345.

\bibitem[Mei, 2006]{Mei2006}
Mei, Y. (2006).
\newblock Sequential change-point detection when unknown parameters are present
  in the pre-change distribution.
\newblock {\em The Annals of Statistics}, 34(1):92--122.

\bibitem[M\'iguez et~al., 2013]{bootstrapfilter2013}
M\'iguez, J., Crisan, D., and Djuri\'c, P.~M. (2013).
\newblock On the convergence of two sequential monte carlo methods for maximum
  a posteriori sequence estimation and stochastic global optimization.
\newblock {\em Statistics and Computing}, 23(1):91--107.

\bibitem[Pesaran et~al., 2004]{Pesaran2004}
Pesaran, M.~H., Pettenuzzo, D., and Timmermann, A. (2004).
\newblock Forecasting time series subject to multiple structural breaks.
\newblock Cambridge Working Papers in Economics 0433, Faculty of Economics,
  University of Cambridge.

\bibitem[Rio, 1995]{Rio:1995}
Rio, E. (1995).
\newblock The functional law of the iterated logarithm for stationary strongly
  mixing sequence.
\newblock {\em The Annals of Probability}, 23:1188--1203.

\bibitem[Stelzer, 2009]{Stelzer2009}
Stelzer, R. (2009).
\newblock On markov-switching arma processes-stationarity, existence of
  moments, and geometric ergodicity.
\newblock {\em Econometric Theory}, 25:43--62.

\bibitem[Tian and Anderson, 2014]{Tian2014}
Tian, J. and Anderson, H.~M. (2014).
\newblock Forecast combinations under structural break uncertainty.
\newblock {\em International Journal of Forecasting}, 30(1):161--175.

\bibitem[Tong, 1978]{Tong1978}
Tong, H. (1978).
\newblock On a threshold model.
\newblock In {\em Pattern Recognition and Signal Processing. NATO ASI Series E:
  Applied Sc.}, pages 575--586. Oxford University Press.

\bibitem[Viterbi, 1967]{Viterbi1967}
Viterbi, A. (1967).
\newblock Error bounds for convolutional codes and an asymptotically opti- mum
  decoding algorithm.
\newblock {\em IEEE Transactions on Information Theory}, 13(2):260--269.

\bibitem[Yau et~al., 2014]{Yau-et-al14}
Yau, C.~Y., Tang, C.~M., and Lee, T. C.~M. (2014).
\newblock Estimation of multiple-regime threshold autoregressive models with
  structural breaks.
\newblock {\em Journal of American Statistical Association},
  110(511):1175--1186.

\bibitem[Zucchini and MacDonald, 2009]{ZucchiniMacDonald2009}
Zucchini, W. and MacDonald, I.~L. (2009).
\newblock {\em {Hidden Markov models for time series: an introduction using
  R}}.
\newblock CRC press, Cheshire, Connecticut.

\end{thebibliography}

\end{document}